\documentclass[11pt]{amsart}
\usepackage{geometry}                % See geometry.pdf to learn the layout options. There are lots.
\usepackage{graphicx}
\usepackage{hyperref}
\usepackage{amssymb}
\usepackage{epstopdf}

\usepackage[T2A,T1]{fontenc}
\DeclareSymbolFont{cyrillic}{T2A}{cmr}{m}{n}
\DeclareMathSymbol{\cL}{\mathalpha}{cyrillic}{203}
\usepackage[toc,page]{appendix}
\newtheorem{theorem}{Theorem}

\newtheorem{conjecture}[theorem]{Conjecture}
\newtheorem{corollary}[theorem]{Corollary}

\newtheorem{definition}[theorem]{Definition}

\newtheorem{lemma}[theorem]{Lemma}

\newtheorem{proposition}[theorem]{Proposition}
\newtheorem{remark}[theorem]{Remark}

\newcommand{\R}{\mathbb{R}}
\newcommand{\Z}{\mathbb{Z}}
\renewcommand{\P}{\mathbf{P}}
\renewcommand{\C}{\mathbb{C}}
\renewcommand{\L}{\mathcal{L}}
\renewcommand{\O}{\mathcal{O}}

\renewcommand{\d}{\mathrm{d}}
\newcommand{\so}{\mathfrak{so}}
\newcommand{\SO}{\mathrm{SO}}
\newcommand{\rO}{\mathrm{O}}
\renewcommand{\sl}{\mathfrak{sl}}
\newcommand{\g}{\mathfrak{g}}

\newcommand{\Sym}{\mathrm{Sym}}
\newcommand{\sd}{\partial}
\newcommand{\Maps}{\mathrm{Maps}}

\newcommand{\Trig}{\mathrm{Trig}}
\newcommand{\Id}{\mathrm{Id}}
\newcommand{\GL}{\mathrm{GL}}

\newcommand{\tGr}{\widetilde{\mathrm{Gr}}}
\newcommand{\rd}{\mathrm{d}}
\newcommand{\ff}{\mathfrak{f}}
\newcommand{\fp}{\mathfrak{p}}
\newcommand{\tr}{\mathrm{tr}}

\newcommand{\e}{b}

\title{A note on  $\sigma$-model with the target  $S^n$}
\author{M.V.Movshev}

\begin{document}
\maketitle
\tableofcontents
\section{Introduction}

It is  well known that the Schr\"odinger operator derived from a general quantum field theory has to be defined on an infinite-dimensional space. For example if the  theory is a sigma-model with  target a sphere \[S^k=\{n\in \mathbb{R}^{k+1}|n\cdot n=R^2\}\] (in physics jargon the theory of $n$-field or $\rO(k+1)$-model) then  the Schr\"odinger operator $H$ has to be defined on the loop space $\L(S^k)=\mathrm{Maps}(S^1,S^k)$. It seems to be a formidable task to define   $H$ acting directly in  $L^2(\L(S^k))$ because of the famous quantum field theory divergencies.

In this paper we investigate the idea  to use instead of $\L(S^k)$ some finite-dimensional approximations  which we denote by $\L_N(S^k)$.
Our hope is that on these  $\L_N(S^k)$ the sigma-model's Schr\"odinger operator $H$ is much better defined and in the limit $N\to \infty$ we will be able to recover spectral property of the full theory.  More precisely, we let 
\begin{equation}\label{D:harmonics}
\Trig_N(\R^{k+1})=\{n( \theta)=v+\sum_{s=1}^Na[s]\cos(s\theta)+b[s]\sin(s\theta)=\sum_{s=0}^Nn_s(\theta)| v,a[s],b[s]\in \mathbb{R}^{k+1}\}
\end{equation} be the space of trigonometric polynomials of degree $N$. $\Trig_N$ contains a subspace  
\begin{equation}\label{E:LNdef}
\L_N(S^k)=\{n(z)\in \Trig_N| n(z)\cdot n(z)=R^2\ \forall z=\exp(\sqrt{-1}\theta)\in \mathbb{C}^{*}\}
\end{equation}
to approximate $\L(S^k)$. It follows immediately form the definition that $n(z)$ is a Laurent polynomial.  The spaces $\L_N(S^k)$ turns to be  Riemannian stratified spaces. For such spaces  definition  of $L^2(\L_N(S^k))$ is straightforward.   The tradeoff is that $\L_N(S^k)$ is singular and this complicates the analysis. 

The spaces $\L_N(S^k)$ are ordered by inclusion:
\[\L_0(S^k)\subset \L_1(S^k)\subset \cdots\]
It turns out that the space of smooth points $\tilde{\L}_N(S^k)$ coincides with $\L_N(S^k)\backslash \L_{N-1}(S^k)$. It is not clear at the moment what boundary conditions on functions $C^{\infty} (\tilde{M}_N(S^k))$ to make the closure of the operator self-adjoint. 

A possible round-about is to remove a tubular neighborhood  $U(\L_{N-1}(S^k))$ in $\L_{N}(S^k)$. The complements $\tilde{\L}_{N,U}(S^k)=   \L_{N}(S^k)\backslash U(\L_{N-1}(S^k))$ is a manifold with a boundary. A Dirichlet or Neumann eigenvalue problem can serve as a replacement to the unknown boundary conditions. The hope is that under certain assumptions in the limit when the size of $U$ goes to zero the spectra $\sigma(H_{\tilde{\L}_{N,U}(S^k)})$ of the operators converge to $\sigma(H_{\tilde{\L}_{N}(S^k)})$ and it is discrete.

One of the fundamental questions posed by physicists  is to establish  the mass gap estimates in the context of sigma models. In our language it would mean that we need to find gap estimates for $\sigma(H_{\tilde{\L}_{N,U}(S^k)})$ and $\sigma(H_{\tilde{\L}_{N}(S^k)})$ when $N\rightarrow \infty$.
\begin{conjecture}(C.f. \cite{DeligneKazhdan} p.1175) Let $\Gamma_N$ be the difference $\lambda_2-\lambda_1$ of the first two eigenvalues of  Neumann (Dirichlet?) problem on $\tilde{\L}_{N,U}(S^k)$. Then there are constant $C_1,C_2>0$ such that \[\Gamma(N)\geq %C\frac{\ln N}{R^2}
C_1{ N^2}e^{-C_2R^2}\]
\end{conjecture}

 It is well known that the sets of eigenvalues $\{\eta_k\}$ of  a positive self-adjoint operator $H_{\tilde{\L}_{N,U}(S^k)}$ can be arranged in a nondecreasing order of magnitude as follows:\[\eta_1\leq \eta_2 \leq\cdots \leq \eta_k \leq\cdots \rightarrow \infty\]
 Much of the work has been done in mathematical community to estimate the difference $\Gamma_k=\eta_k-\eta_1$  for the Neumann problem for a general Schr\"{o}dinger operator in terms of differential geometry of the underlying manifold \cite{Chen},\cite{Chen1990},\cite{WangP}. 
 
Here is some relevant definitions: let $(M^m, g_{ij})$ be an $m$-dimensional compact Riemannian manifold with metric $g_{ij}$ and $\sd M \neq  \emptyset$ be the boundary of $M$. Let $\Delta$ be the Laplacian operator associated to $g_{ij}$ on $M$ then we define on $M$ a Schr\"{o}dinger operator by $\Delta - q(x)$, where $q(x) \in  C^2(M)$. We consider the following Neumann
eigenvalue problem: 
 \[\begin{split}
& 	\Delta u-qu=-\eta u \text{ in } M\\
& \frac{\sd u}{\sd \nu} = 0	 \text{ on } \sd M.
\end{split}
\]

\begin{definition}{\rm 
$\sd M$ is said to satisfy the "interior rolling $R$-ball" condition if for each point $p\in  \sd M$ there is a geodesic ball $B_q(\frac{R}{2})$, centered at $
q\in M$ with radius $R/2$,such that $p=B_q(\frac{R}{2})\cap\sd M$ and $B_q(\frac{R}{2})\subset M$.
}\end{definition}
\begin{theorem}\cite{Chen}\label{E:ch1}	Let $M^m$, $m \geq 3$, be an $m$-dimensional compact manifold with boundary $\sd M$ and $\omega(q) = \sup q - \inf q$ denote the oscillation of $q$ over $M$. Suppose that the Ricci curvature of $M$ satisfies $R_{ij} +Kg_{ij}$ is positive-definite, $K \geq 0$ and the second fundamental form $\alpha_{ij}$ of $\sd M$ with respect to outward pointing unit normal $\nu$ satisfies $\alpha +Hg_{ij}|_{\sd M}$ positive-definite $ H \geq 0$. Suppose that $\sd M$ also satisfies the "interior rolling $r$-ball condition" with $r$ chosen small (see \cite{Chen1990}, \cite{WangP} for the choice of r). Then the gap $\Gamma_k$ of $k$-th Neumann eigenvalues $\eta_k$ and $\eta_1$ of $M$ satisfies
$\Gamma_k \geq \alpha_1k^{\frac{2}{m}}$
for all $k = 2,3,\dots$ and some explicitly computable constant $\alpha_1$ depending on $m, K, H, r$ diameter of $M$ and $\omega(q)$.
\end{theorem}
Here are some  formulations for the problem with zero potential.
\begin{theorem}\cite{Chen1990}{\rm
Let $M$ be as in Theorem \ref{E:ch1}. Then the first nonzero Neumann
eigenvalue $\lambda_1$ of Laplacian operator has a lower bound given by

\[\lambda_1\geq C_p=\frac{1}{1+H^2}\left(\frac{1-\alpha^2}{4(m-1)d^2}B_1^2-B_2\right)\exp(-B_1)\]
$\alpha$and $R$ are positive constants less then one,
\[d=\text{ diameter of } M\]
\[B_1=1+\left(1+\frac{4(m-1)d^2B_2}{1-\alpha^2}\right)^{\frac{1}{2}}\]
\[B_2=(1+H)B_3+\frac{((2m-3)^2+(4m-5)\alpha^2)H^2}{(m-1)R^2\alpha^2}+(1+H)^2K\]
\[B_3=\frac{2(m-1)H(1+H)(1+3H)}{R}+\frac{H(1+H)}{R^2}\]
}\end{theorem}
\begin{theorem}\label{T:M}
\cite{Meyer}{\rm \
Let $n\geq 2$ be an integer and let $D,a,e,K$ be real numbers with $D>0$ and $a>0$. Then a continuous function $C(n,e,a,D,K)$ of these numbers can be constructed, such that for any compact, Riemannian manifold $(M,\sd M,g)$ with diameter $d$ the inequality \[\lambda_1 d^2\geq C(n,e,a,K,D) \]holds, provided that the following conditions are satisfied: 
\begin{enumerate}
\item $d\leq D$, 
\item  the Ricci curvature is $\geq (n-1)K$, 
\item  the so-called "sectional radius'' (defined as the radius of injectivity of the exponential map of the normal bundle of $\sd M$ into $M$ is $\geq a$, and 
\item  the principal curvatures of the boundary are $\leq e$.
Where
\[C(n,e,a,K,D)=\alpha(n)\exp\left( -\beta(n)D\max\left(e,\frac{1}{a}\right)\right)\text{ if } K\geq 0\]
\[C(n,e,a,K,D)=\alpha(n)\exp\left( -\beta(n)D\max\left(e,\frac{1}{a}\right)-\gamma(n)D\frac{\sqrt{|K|}}{Th\left(\frac{a\sqrt{|K|}}{2}\right)}\right) \text{ if } K< 0\]
with
\[\alpha(n)=\frac{e^{-2}}{4(n-1)}\]
\[\beta(n)=16n\]
\[\gamma(n)=2(n-1)^{\frac{3}{2}}\]
\end{enumerate}
}\end{theorem}
This theorem motivates our interest in differential geometry of $\tilde{\L}_{N,U}(S^k)$. In this paper (Section \ref{S:ricci}) we are going to present explicit  formulas for Ricci curvature of $\tilde{\L}_{N,U}(S^k)$ the second quadratic form $\alpha$.

We also going to study (Section \ref{S:Schr}) Schr\"odinger equation on the space $\L_1(S^k)$. We derive asymptotic of radial eigenfunctions  near singular locus of $\L_1(S^k)$.

\section{The classical $\rO(k+1)$-sigma model}
In this section we are goin to remind the reader the basic definition related to the theory of n-field. First of all the field of the theory  is the smooth map $n:\R\times S^1\to \R^{k+1}$ subject to a constraint. To formulate the constraint we choose coordinates $t,\theta$ on $\R\times S^1$. Also we equip $\R^{k+1}$ with the dot-product such that for $x,y\in \R^{k+1}, x\cdot y=\sum_{i=1}^{n+1}x^iy^i$.  We will usually use an abbreviation $x^2:=x\cdot x$. We assume that $n$ satisfies
\begin{equation}\label{E:constraint}
n(t,\theta)^2-R^2=0
\end{equation}
\begin{equation}\label{E:period}
n(t,\theta+L)=n(t,\theta)
\end{equation}

Thus, effectively, $n$ is a map $\R\times S^1\to S^k$, where $S^1$ is a circle of radius $\frac{L}{2\pi}$ and $S^k$ is a sphere of radius $R$.
Lagrangian $L(n)\d t \d \theta$ of the theory is 
\begin{equation}\label{E:lag0}
L(n)\d t \d \theta=\frac{1}{2}\left(n_t^2-n_{\theta}^2\right)\d t \d \theta
\end{equation}
In our finite approximation the space of fields $\Maps(\R\times S^1, S^k)=\Maps(\R, \L(S^k))$ gets replaced by $\Maps(\R, \L_N(S^k))$.
The circle of questions typically discussed when Lagrangian has been written is what are the solutions of the equation of motion and what are spectral properties of the associated Schr\"{o}dinger operator. In the present situation we have to postpone answering these questions there is more pressing one: what are the singularities of  $\L_N(S^k)$ and is it possible to resolve them. We will answer these question in the next sections

\section{Resolution of singularities}
$\L(S^k)$ is a homogenous space of the loop groups $\L(\SO(k+1))\subset \L(\rO(k+1))$. It means that we have the action map 
\begin{equation}\label{E:faction}
\L(\rO(k+1))\times \L(S^k)\rightarrow  \L(S^k)
\end{equation}
and the product map
\begin{equation}\label{E:product}
\L(\rO(k+1))\times \L(\rO(k+1))\rightarrow \L(\rO(k+1))
\end{equation}
Only the action of constants $\rO(k+1)\subset \L(\rO(k+1))$ survives after the truncation $\L(S^k)\Rightarrow \L_{K}(S^k)$. The group $\rO(k+1)$ is an algebraic submanifold in $Mat_{k+1}(\R)$. Following the analogy with $S^k$ we define a finite-dimensional  subvariety $\L_{M}(\rO(k+1))\subset \L(\rO(k+1))$. The maps
\begin{equation}\label{E:Paction}\begin{split}
&\L_{M}(\rO(k+1))\times \L_{K}(S^k)\rightarrow  \L_{M+K}(S^k)\\
&\L_{M}(\rO(k+1))\times \L_{K}(\rO(k+1))\rightarrow \L_{M+K}(\rO(k+1)) \end{split}
\end{equation} are   restriction of (\ref{E:faction}, \ref{E:product}).

The set of (smooth) homomorphisms $\mathrm{Hom}(S^1,\SO(k+1))$ is a subset of $\L(\rO(k+1))$ . The group $\SO(k+1)$ acts on $\mathrm{Hom}(S^1,\SO(k+1))$ by conjugations and $\mathrm{Hom}(S^1,\SO(k+1))$ is a union of $\SO(k+1)$-orbits. The intersection $\mathrm{Hom}(S^1,\SO(k+1))\cap \L(\SO_1(k+1))$ contains an element $\lambda_0$. It is a block-sum of standard two-dimensional representation $r$ and a trivial $k-1$-dimensional representation. We denote the $\SO(k+1)$ orbit of $\lambda_0$ by $G(k+1)$.

The map 
\begin{equation}\label{E:decomp}
\mu_N:G(k+1)\times \L_{N-1}(S^k)\rightarrow  \L_{N}(S^k)
\end{equation}
is a restriction of the top map in (\ref{E:Paction}) with $M=1$ and $K=N-1$ from $\L_{1}(\rO(k+1))$ to $G(k+1)$.

The goal of this section is to construct an open subset \[X_N\subset G(k+1)\times (\L_N(S^{k+1})\backslash \L_{N-1}(S^{k+1}))\] which is isomorphic to $\L_{N+1}(S^{k+1})\backslash \L_{N-1}(S^{k+1})$ and it projects onto $ \L_N(S^{k+1})\backslash \L_{N-1}(S^{k+1}),N\geq 1$. The fibers of the projection are affine spaces  . Thus  $\mu_N$ is binational isomorphism, We can use $\mu_N$ for different $N$ to birationally  identify $\L_{N}(S^k)$ with $G(k+1)^{\times N}\times S^k$.  Moreover the map  $\mu:G(k+1)^{\times N}\times S^k\to \L_N(S^{k+1})$ ($\mu(\lambda_N(z),\dots,\lambda_1(z),n)=\lambda_N(z)\cdots\lambda_1(z)n$ ) is a resolution of singularities.

The reader might wish to compare this with a ideologically similar result \cite{Gruppy} Appendix A. in which the  authors describe generators and relations in the group polynomial loops in $\mathrm{U}(n)$.

First we want to identify the space $G(k+1)$.
Representation $\lambda\in G(k+1)$ carries an information about a two-dimensional subspace $V(1)\subset \R^{k+1}$ together with the orientation $or$ specified by the operator $W=\lambda(\sqrt{-1}\frac{\pi}{2})$. Conversely, any point  $V^{or}$ in an  oriented Grassmannian $\tGr_2(\R^{k+1})$ determins a homomorphism $\lambda_{V^{or}}\in G(k+1)$ by the rule
\begin{equation}\label{E:rot}
\begin{split}
&\lambda_{V^{or}}(\theta)=P_{V}\cos(\theta) +W_{V}\sin(\theta)+\Id-P_{V}\\
&P_{V}^2=P_{V},W_{V}^2=-P_{V},  \mathrm{Im} P_{V}=V
\end{split}
\end{equation}
where $P_{V}$ is an orthogonal projection on $V$, $W_{V}$ is an operator of projection on $V$ followed by rotation on ninety degree in agreement with orientation $or$.  Note that $\lambda_{V^{-or}}(\theta)=\lambda_{V^{or}}(\theta)^{-1}$ and $G(k+1)$ is closed under taking inverses.

For construction of $X_N$ we will use some lemmas.

\begin{lemma}\label{L:decr}
Let  
\begin{equation}\label{E:pairofvectors}
\{a,b|||a||=||b||\}
\end{equation} be  an orthogonal basis for $V^{or}\in \tGr_2(\R^{k+1})$. In the future we will call an admissible basis  an orthogonal  basis that satisfies \ref{E:pairofvectors}. The map  $\gamma_l:S^1\rightarrow V,\gamma_l(\theta)= \cos(l\theta)a+\sin(l\theta) b$ satisfies \[\lambda_{V^{-or}}\gamma_l=\gamma_{l-1}\] if $\{a,b\}$ is a basis of $V$ that agrees with orientation and 
 \[\lambda_{V^{or}}\gamma_l=\gamma_{l-1}\] if $\{a,b\}$ has the opposite orientation.
\end{lemma}
\begin{lemma}\label{L:hom}
Let $V^{or}$ be a space generated by a positively oriented basis $\{a[N],b[N]\}$-the leading coefficients (\ref{E:ab}) of a fixed $n(\theta)\in \L_N(S^{k})\backslash \L_{N-1}(S^{k})$.
Let $P$
 be the orthogonal projection on $V$. Vectors $P(a[N-1]), P(b[N-1])$ have equal lengths and if nonzero are orthogonal  and define the same orientation as $\{a[N],b[N]\}$.
\end{lemma}
\begin{proof}
Nonzero vectors $a[N],b[N]$ are orthogonal and $a[N]^2=b[N]^2$ (formulas $d_{2N}, c_{2N}$, equation \ref{E:defeq}). The formula for projection is \[P(v)=\frac{a[N]\cdot v}{a[N]^2} a[N]+\frac{b[N]\cdot v}{b[N]^2} b[N]\] Then
\[P(a[N-1])\cdot P(b[N-1])=\frac{a[N]\cdot a[N-1]}{a[N]^2} \frac{a[N]\cdot b[N-1]}{a[N]^2}+\frac{b[N]\cdot a[N-1]}{b[N]^2} \frac{b[N]\cdot b[N-1]}{b[N]^2}=\]
\[=\frac{1}{a[N]^2a[N]^2}\left((a[N]\cdot a[N-1])(a[N]\cdot b[N-1])+(b[N]\cdot a[N-1])(b[N]\cdot b[N-1])\right)=\]
\[=\frac{a[N]\cdot a[N-1]}{a[N]^2a[N]^2}\left(a[N]\cdot b[N-1]+b[N]\cdot a[N-1]\right)=0\]
We use formulas $d_{2N-1}$ and $c_{2N-1}$ from  \ref{E:defeq}:\[a[N]\cdot b[N-1]+b[N]\cdot a[N-1]=0, a[N]\cdot a[N-1]-b[N]\cdot b[N-1]=0.\]
A similar computation shows that quantities
\[P(a[N-1])^2=(a[N]\cdot a[N-1])^2+(b[N]\cdot a[N-1])^2\]
\[P(b[N-1])^2=(a[N]\cdot b[N-1])^2+(b[N]\cdot b[N-1])^2\]
are equal.
Transition matrix   from $\{a[N],b[N]\}$ to $\{P(a[N-1]),P(b[N-1])\}$  
\[ \left( \begin{array}{cc} 
\frac{a[N]\cdot a[N-1]}{a[N]^2} & \frac{a[N]\cdot b[N-1]}{a[N]^2}  \\
\frac{b[N]\cdot a[N-1]}{b[N]^2} &  \frac{b[N]\cdot b[N-1]}{b[N]^2}  \\ 
 \end{array} \right)\]
 has determinant 
 \[\frac{a[N]\cdot a[N-1]}{a[N]^2}\frac{b[N]\cdot b[N-1]}{b[N]^2}-\frac{b[N]\cdot a[N-1]}{b[N]^2} \frac{a[N]\cdot b[N-1]}{a[N]^2}=\]
 \[=\frac{1}{a[N]^2a[N]^2}\left((a[N]\cdot a[N-1])^2+(a[N]\cdot b[N-1])^2\right)=\frac{P(a[N-1])^2}{a[N]^2a[N]^2}\geq 0\]
This  verifies the statement about orientation.
\end{proof}
\begin{proposition}\label{P:onto}
The map (\ref{E:decomp}) is   onto.
\end{proposition}
\begin{proof}
It suffices to show that for any $N\geq 1$ and  $n(\theta)\in \L_N(S^{k})\backslash \L_{N-1}(S^{k}) \exists \lambda\in G(k+1)$ such that 
$\lambda^{-1}n\in  \L_{N-1}(S^{k})$. Then trivially $\lambda(\lambda^{-1}n)=n$.
Then the  general statement will follows because $\L_N(S^{k})$ is  filtered by $ \L_l(S^{k})\backslash \L_{l-1}(S^{k}),l=0,\dots,N$.

We denote a homomorphism $\lambda_{V^{or}}$ based on the space $V$ from Lemma \ref{L:hom} by $\lambda$. We would like to show that  
$\lambda^{-1}n\in \Trig_{N-1}$. In terms of  Fourier harmonics $n=\sum_{l=0}^Nn_l$ (\ref{D:harmonics}) it suffice to verify that 
$\lambda^{-1}n_N, \lambda^{-1}n_{N-1}\in \Trig_{N-1}$. All other terms $\lambda^{-1}n_l, l\leq N-2$ belong to $ \Trig_{N-1}$ by virtue of basic trigonometric identities. Note that $\lambda^{-1}n_N\in \Trig_{N-1}$ follows  from Lemma \ref{L:decr} and the choice of $\lambda$. Denote $P_V$ from Lemma \ref{L:hom} by $P$. Then $\lambda^{-1}$ acts trivially on $(\Id-P)n_{N-1}$ so
$(\Id-P)n_{N-1}\in \Trig_{N-1}$. By Lemma  \ref{L:hom} $\lambda^{-1}Pn_{N-1}\in \Trig_{N-2}$, thus $\lambda^{-1}n_{N-1}\in  \Trig_{N-1}$ and Proposition is proven.

The set $\L_{1}(\SO(k+1))$ is invariant with respect to taking inverses because $\phi(\theta)^{-1}=\phi(\theta)^{t}$, where $t$ stands for transpose which is a linear operation. Then $n(\theta)$ is equal to $\phi(\theta)^{-1}\phi(\theta)n(\theta)$. The statement would follow if for any $n(\theta)\in \L_{N}(S^k)$ we will find $\phi(\theta)\in \L_{1}(\SO(k+1))$ such that $\phi(\theta)n(\theta)\in \L_{N-1}(S^k)$.
Matrix coefficients of  $\phi(\theta)=V+A\cos(\theta)+B\sin(\theta)\in \L_{1}(\SO(k+1))$ have to satisfy
\begin{equation}\label{E:loopone}\begin{split}
&AB^t+BA^t=0,AA^t-BB^t=0\\
&VB^t+BV^t=0, VA^t+AV^t=0\\
&AA^t+VV^t=Id
\end{split}\end{equation}
We expand  $\phi(\theta)n(\theta)$ into Fourier series $a'[N+1]\cos((n+1)\theta)+b'[N+1]\sin((n+1)\theta)+a'[N]\cos(n\theta)+b'[N]\sin(n\theta)+\cdots$ whose coefficients satisfy
\begin{equation}\label{E:loopaction}\begin{split}
&a'[N+1]=\frac{1}{2}(Aa[N]-Bb[N]),b'[N+1]=\frac{1}{2}(Ba[N]+Ab[N])\\
&a'[N]=\frac{1}{2}(Aa[N-1]-Bb[N-1]+2Va[N]),b'[N]=\frac{1}{2}(Ba[N-1]+Ab[N-1]+2Vb[N])\\
&\cdots
\end{split}\end{equation}
Our goal is to find $A,B,V$ that satisfy constraints (\ref{E:loopone}) such that $a'[N+1],b'[N+1],a'[N],b'[N]$ vanish.	
Introduce notations: $(a,b)_{ij}$ stands for a matrix whose $i$-th row is a vector $a$,  $j$-th row is a vector $b$ and all the remaining rows vanish. From equations (\ref{E:defeq}) we know that if one of $a[N],b[N]$ is nonzero then both are nonzero. Let $a=a[N]/|a[N]|$,  $b=b[N]/|b[N]|$ be normalized leading coefficients. We set $A=(a,b)_{k,k+1}$, $B=(b,-a)_{k,k+1}$. To define $V$ we choose $k-1$ orthonormal vectors $e_1,\dots,e_{k-1}$ perpendicular to $a$ and $b$. We define $V$ to be the matrix $(e_1,\dots,e_{k-1},0,0)$, whose first $k-1$ are formed by vectors $e_i$ and whose remaining two rows are equal to zero. In geometric terms $V$ is an orthogonal projection on orthogonal complement $<a,b>^{\perp}$. We leave it to the reader to verify that equations (\ref{E:loopaction},\ref{E:loopone}). The reader should use four equations from the set (\ref{E:defeq}) $a[N]^2=b[N]^2, a[N]\cdot b[N]=0,a[N]\cdot a[N-1]-b[N]\cdot b[N-1]=0, a[N]\cdot b[N-1]+b[N]\cdot a[N-1]=0$. We denote the resulting element in $\L_1(\SO(k+1))$ by  $\phi_{ab}$

\end{proof}

\begin{proposition}\label{E:homdec}
Pick $n\in \L_{N}(S^{k})\backslash \L_{N-1}(S^{k})$. 
Suppose $n=\lambda_in_i, \lambda_i\in G(k+1), n_i\in \L_{N-1}(S^k),i=1,2$. Then $\lambda_1=\lambda_2$ and $n_1=n_2\in  \L_{N-1}(S^{k})\backslash \L_{N-2}(S^{k})$.

\end{proposition}
\begin{proof}

 We decompose maps into finite Fourier series 
\[\begin{split}
&n=v+\sum_{k=1}^Na[k]\cos(k\theta)+b[k]\sin(k\theta),\quad n_i=v^i+\sum_{k=1}^Na^i[k]\cos(k\theta)+b^i[k]\sin(k\theta),i=1,2\\
&\lambda_i=P_i\cos(\theta) +W_i\sin(\theta)+\Id-P_i,i=1,2
\end{split}
\] 
Then
\begin{equation}\label{E:laction}
\begin{split}
&a[N+1]=P_ia^i[N]-W_ib^i[N]\\
&b[N+1]=W_ia^i[N]+P_ib^i[N],i=1,2
\end{split}
\end{equation}
By assumptions $a[N+1],b[N+1]\neq0$ and by virtue of equations (\ref{E:rot}) $W_ia[N+1]=b[N+1], W_ib[N+1]=-a[N+1], i=1,2$. This implies that \[<a[N+1],b[N+1]>\subset \mathrm{Im}P_i+\mathrm{Im}W_i\subset \mathrm{Im}P_i\Rightarrow <a[N+1],b[N+1]>= \mathrm{Im}P_i\] and because of that $P_1=P_2,W_1=W_2$. We conclude that $\lambda_1=\lambda_2=\lambda$ and $\lambda^{-1}n=n_1=n_2$. Were $n_1$ an element of $\L_{N-2}(S^k)$, then $\lambda n_1$ would be an an element of $\L_{N-1}(S^k)$, which is impossible because $\lambda n_1=n= \L_{N}(S^k)\backslash  \L_{N-1}(S^k)$.
\end{proof}

\begin{remark}
The space $G(k+1)$ is diffeomorphic to a complex quadric $Q\subset \P(\C^{k+1})$: an admissible  basis (\ref{E:pairofvectors}) for $V^{or}\in  \tGr_2(\R^{k+1})$ generate a line  $<a+\sqrt{-1}b>\in \P(\C^{k+1})$.  Equation $(a+\sqrt{-1}b,a+\sqrt{-1}b)=||a||^2-||b||^2+2\sqrt{-1}(a,b)=0$ is automatically satisfied. Complex rescaling of the generator $a+\sqrt{-1}b$ corresponds to different choices of admissible bases for $V^{or}$
\end{remark}

Now we are ready to determine the subset $X_N$ 

\begin{definition}
Fix $n\in \L_{N}(S^k)$. 
We call an element  $\lambda\in G(k+1)$ $n$-regular if $\lambda n\in  \L_{N+1}(S^k)\backslash \L_{N}(S^k)$. If $\lambda n\in  \L_{N}(S^k)$ then we call $\lambda$ an $n$-singular homomorphism. We denote the set of $n$-singular elements in $G(k+1)$ by $D(n)$.
\end{definition}
\begin{proposition}
The set $G(k+1)\backslash D(n)$ is algebraically isomorphic to $\C^{k-1}\cong \R^{2k-2}$.
\end{proposition}
\begin{proof}
Equation $\lambda n\in \L_{N}(S^k)$ is equivalent to equations
\begin{equation}\label{E:degenerate}\begin{split}
&Pa[N]-Wb[N]=0\\
&Wa[N]+Pb[N]=0
\end{split}\end{equation}
 on the Fourier coefficients of the highest degree.
Representation  $\lambda$ is defined in (\ref{E:rot}). We pick an admissible  basis $a,b\in \mathrm{Im} P$. Operators $P$ and $W$ can be defined purely in terms of the basis:
\[\begin{split}
&P(v)=\frac{1}{||a||^2}\left( (a\cdot v)a+(b\cdot v) b\right)\\
&W(v)=\frac{1}{||a||^2}\left( -  (a\cdot v)b+(b\cdot v)a \right)
\end{split}\]
Equations (\ref{E:degenerate}) are equivalent to
\[\begin{split}
&a\cdot a[N]=b\cdot b[N]\\
&b\cdot a[N]=-a\cdot b[N]
\end{split}\]
These equations are equivalent to orthogonality of complex vectors
\[A=a[N]+\sqrt{-1}b[N]\]
\[A'=(a+\sqrt{-1}b)\]
\[A\cdot A'=\left(a[N]\cdot a - b[N]\cdot b\right) +\sqrt{-1}(a[N]\cdot b + b[N]\cdot a)=0\]
Together with admissible basis equations $A\cdot A=A'\cdot A'=0$ these ensure  that a $\P^1$ spanned by $A,A'$ lies entirely  in $Q\subset \P(\C^{k+1})$. 

For fixed $A\in Q$ the set $D=\{X\in Q|A\cdot X=0\}$ is invariant under the parabolic subgroup  $P\subset \SO(k+1,\C)$ that fixes $A$.  Thus $D$ is a codimension one  Schubert cell in $Q$. The complement of $D$ is isomorphic to affine space $\C^{k-1}$.
Suppose $k+1=2l$ is even. We can choose a system of complex coordinates such that quadric equation is $\sum_{i=1}^lz_iw_i=0$, the set $X$ is defined by equation $z_1=0$. We can solve quadric equation in the complement of $D$ \[w_1=-\frac{1}{z_1}\sum_{i=2}^lz_iw_i\] This verifies isomorphism of the space of solutions with $\C^{k-1}$.

Suppose $k+1=2l+1$ is odd. The quadric equation in suitable coordinates is $h^2+\sum_{i=1}^lz_iw_i=0$. The set $X$ is still defined by equation $z_1=0$.  \[w_1=-\frac{1}{z_1}\left(h^2+\sum_{i=2}^lz_iw_i\right)\]

To summarize-the set $Q\backslash D$ is isomorphic to $G(k+1)\backslash D(n)$ is contractible
 %\ref{E:laction}
\end{proof}
We arrive at the following conclusion
\begin{corollary}
The space $\L_{N+1}(S^{k+1})\backslash \L_{N}(S^{k+1})$ is homotopy equivalent to $\L_N(S^{k+1})\backslash \L_{N-1}(S^{k+1})$, $N\geq 1$.  Thus by results in Section \ref{E:smalldeg} $\L_{N+1}(S^{k+1})\backslash \L_{N}(S^{k+1})$ is homotopy equivalent to $\L_{1}(S^{k+1})\backslash \L_{0}(S^{k+1})$, which is homotopy equivalent to the  Stiefel  manifold $V_2$.%circle bundle over $ \tGr_2(\R^{k+1})=Q^{k-1}$.
\end{corollary}
\begin{remark}
The spaces $\L_N(S^3)$ are dense in   $\L(S^3)$  in $L^{\infty}$ topology. For the proof see David E Speyer response at https://mathoverflow.net/questions/158105/real-varieties-with-enough-algebraic-loops
\end{remark}

\section{Equation of motion}
The Lagrangian (\ref{E:lag0})
with $n(t,\theta)\in \Maps(\R,\L_N(S^k))$ 
can be written in terms of unconstrained fields by means of the Lagrange multiplier $c$:
 \begin{equation}\label{E:lag1}
 \L(n,c):=\frac{1}{2}\left(n_{t}\cdot n_{t}-n_{\theta}\cdot n_{\theta}+c(n\cdot n-R^2)\right), \quad n\in \Trig_{N}(\R^{k+1}), c\in \Trig_{2N}(\R)
 \end{equation}
Introduce a notation 
\begin{equation}\label{E:dirichlet3}
D_N(\theta)=\frac{\sin((N+1/2)\theta)}{\sin(\theta/2)}
\end{equation}

Denote by $(\Pr_Nn)(\theta)=\frac{1}{2\pi}\int_{S^1}D_N(\theta-\theta')n(\theta')\rd \theta'$ the orthogonal projection onto $\Trig_{N}(\R^{k+1})\subset \Trig_{N+N'}(\R^{k+1})$.

Equations of motion for (\ref{E:lag1}) look like
\begin{equation}\label{E:em1N}
{\Pr}_N(-n_{tt}+n_{\theta\theta}+cn)=0
\end{equation}
\begin{equation}\label{E:constrN}
n\cdot n-R^2=0
\end{equation}

 Note that ${\Pr}_N(-n_{tt}+n_{\theta\theta})=-n_{tt}+n_{\theta\theta}$, so (\ref{E:em1N}) is equivalent to 
\[-n_{tt}(t,\theta)+n_{\theta\theta}(t,\theta)+\frac{1}{2\pi}\int_{0}^{2\pi}c(t,\theta')n(t,\theta')D_N(\theta'-\theta)d\theta'=0
\] After taking the dot-product with $n$ we get
\[\begin{split}
&n_{t}^2-n^2_{\theta}+g_{n,N}(c)=0
\end{split}\]
Here $g_{n,N}$ is an operator on $\Trig_{2N}(\R)$ define by the formula
 \[g_{n,N}(c)=\frac{1}{2\pi}\int_{0}^{2\pi}c(t,\theta')n(t,\theta')\cdot n(t,\theta)D_N(\theta'-\theta)d\theta'\]
The inverse to $g_N$ is the operator $G_{n,N}$.
  we can solve the last equation for $c$:
\[c=G_{n,N}(n^2_{\theta}-n_{t}^2)\]
finally equation (\ref{E:em1N}) becomes equation of extremals
\begin{equation}\label{E:EMfin}
-n_{tt}(t,\theta)+n_{\theta\theta}(t,\theta)+\frac{1}{2\pi}\int_{0}^{2\pi}G_N(n^2_{\theta}-n_{t}^2)(t,\theta')n(t,\theta')D_N(\theta'-\theta)d\theta'=0
\end{equation}
The remarkable fact \cite{EF} is that the $\rO(k+1)$-model is completely integrable. It is not clear if it admits a sequence of finite integrable approximations. In particular it is not clear whether (\ref{E:EMfin}) is integrable.

If we explicitely impose the constraints (\ref{E:constrN}) and $n\in \Trig_{N}(\R^{k+1})$ we get the action for  a system that moves on a manifold $\L_N(S^k)$ in the potential $U=\int_{S^1}n_{\theta}^2$. On general grounds we know that classical trajectories satisfy  Newton's law equation $\nabla n_t'=-*dU(n)$. $\nabla$ is the Levi-Civita connection  associated with the metric on $\L_N(S^k)$. In the next section we will describe the metric on $\L_1(S^k)$-the simplest manifold in the family $\L_N(S^k)$.

\section{Variety $\L_1(S^k)$}\label{E:smalldeg}
The space $\L_1(S^k)$ is the simplest nontrivial case of our construction. It consists of a triple of vectors $(v,a,b)=(v,a[1],b[1])$ 
that satisfy 
\begin{equation}\label{E:T1equations}
\begin{split}
&v\cdot a=v\cdot b=a\cdot b=0\\
&a^2-b^2=0\quad v^2+a^2-R^2=0
\end{split}
\end{equation}
This data has a simple geometric interpretation (cf. \cite{WKlingenberg} Section 2.3). Nonzero vectors $a$ and $b$ span a two-plane $W_{a,b}\subset \R^{k+1}$, which   intersects  $S^k$ by a big circle. We see that the  Grassmannian $\tGr_2(\R^{k+1})$ coincides with the space of such circles. A choice of an  orthogonal basis $\{a,b\}$ for $W$ encodes a parametrization of $W\cap S^k$. A shift $W\rightarrow v+W$ of $W,v\in W^{\perp}$ produces a family of smaller circles $(W+v)\cap S^k$, which degenerate to points of $S^k$ when $v\cdot v=R^2$. Thus for a fixed $W$  the set of admissible shifts is a $k-1$-dimensional disk $\overline{D}^{k-1}=\{v\in W^{\perp}|v\cdot v\leq R^2\}$. From this we deduce that  the complement $\L_1(S^k)\backslash \L_0(S^k)$ is a $D^{k-1}\times S^1$-bundle over oriented Grassmannian $\tGr_2(\R^{k+1})$: under projection
\[p:\L_1(S^k)\backslash \L_0(S^k)\rightarrow \tGr_2(\R^{k+1})\]
the triple $(v,a,b)$ get mapped to $W_{a,b}\in \tGr_2(\R^{k+1})$. The following two vector bundles  over $\tGr_2(\R^{k+1})$ were extensively studied in the theory of characteristic classes (cf \cite{MilnorStasheff}):
\[\gamma_2=\{(v,W)|v\in W,W\in \tGr_2(\R^{k+1})\}\rightarrow \{W|W\in \tGr_2(\R^{k+1})\}\]
\[\gamma^{\perp}_2=\{(v,W)|v\in W^{\perp},W\in \tGr_2(\R^{k+1})\}\rightarrow \{W|W\in \tGr_2(\R^{k+1})\}\]

Our  bundle is a fiber product of the disk and the circle bundles $D^{k-1}(\gamma^{\perp}_2)\times S(\gamma_2)$.

It is convenient to  introduce vectors $e_v=\frac{v}{|v|}, e_a=\frac{a}{|a|},e_b=\frac{b}{|b|}$, A triple $(e_v,e_a,e_b)$ defines a point in a  Stiefel manifold $V_3(\R^{k+1})$ of  orthonormal 3-frames in $\R^{k+1}$.   
In this paper we will use two local parametrizations:
\begin{equation}\label{E:map}
\begin{split}
& trig:(0,\pi R/2) \times V_3(\R^{k+1}) \rightarrow \L_1(S^k)\\
& alg:(0,1) \times V_3(\R^{k+1}) \rightarrow \L_1(S^k).
\end{split}
\end{equation}
They are defined by the formula:
\begin{equation}\label{E:mapformula}
\begin{split}
&trig(\tau,e_v,e_a,e_b)=\left( R\sin\left(\tau/R\right) e_v, R\cos\left(\tau/R\right)e_a,R\cos\left(\tau/R\right)e_b\right) \\
&alg(t,e_v,e_a,e_b)=\left( Rt^{1/2} e_v, R(1-t)^{1/2}e_a,R(1-t)^{1/2}e_b\right) 
\end{split}
\end{equation}

The inverses  
\[\begin{split}
&trig^{-1}(v,a,b)=\left(R\arcsin \left(|v|/R\right), v/|v| ,  a/|a| , b/|b| \right) \\
&alg^{-1}(v,a,b)=\left(|v|/R, v/|v| ,  a/|a| , b/|b| \right) \end{split}\]
are defined away from the sets $S= \{(Re_v,0,0)|e_v\in S^{k}\}$ and $S'=\{(0,Re_a,Re_b)|(e_a,e_b)\in V_2(\R^{k+1})\}$

\begin{proposition}
The singular  locus $\L_1(S^k)^{sing}$ coincides with $S$.
\end{proposition}
\begin{proof}
The set $S$ is a subset of $\L_1(S^k)^{sing}$ because the differentials  $d(a^2-b^2)$ and $d(a\cdot b)$ vanish on $S$. Let  $S'$ be $\{(0,Re_a,Re_b)|(e_a,e_b)\in V_2(\R^{k+1})\}$. The complement $\L_1(S^k)\backslash (S\cup S')$ is smooth because the map  $alg$ defines a diffeomorphism of the set with a smooth manifold.

A direct inspection shows that differentials of (\ref{E:T1equations}) are linearly independent along $S'$. By inverse function theorem $\L_1(S^k)$ is smooth near $S'$.
\end{proof}
As a corollary we see that   $\dim \L_1(S^k)=3k-2$

The space $\Trig_1(\R^{k+1})$ is equipped with   metrics 
\begin{equation}\label{E:flatmetric}
\begin{split}
&g=dv\cdot dv+\frac{1}{2}(da\cdot da+db\cdot db).\\
\end{split}
\end{equation}
Our next goal is to compute restrictions $g=pol^{*}g'$ .%$g_h=pol^{*}g'_h$.

For this purpose we need to have a description of a tangent space $T_x$ to a point 
$x\in  V_3(\R^{k+1})$.  $V_3(\R^{k+1})\cong \SO(k+1)/\SO(k-2)$ is a homogeneous $\SO(k+1)\times \SO(3)$-space. Suppose the point $x$ is represented by a coset $\SO(k-2)$. Then the  tangent space to $x$ is $\so(k+1)/\so(k-2)$. Using the isomorphism $\so(k+1)\cong \Lambda^2 \mathbb{R}^{k+1}$ and the inner product we identify the tangent space to $x=( e_v,e_a,e_b)$ with 
 $\Lambda^2\langle e_v,e_a,e_b\rangle+ \langle e_v,e_a,e_b\rangle\otimes \langle e_v,e_a,e_b\rangle^{\perp}$.
 The space $T_x$ being a subspace in $\R^{k+1}\otimes \R^{k+1}$ carries its own inner product  induced from the metric  $(l\otimes m,l'\otimes m')=(l\cdot l') \times (m\cdot m')$. It is convenient to use the corresponding  metric, which we denote by $h$,   as a reference point in the space of metrics on $V_3(\R^{k+1})$.

We will need a  description of an $h$-orthonormal basis for $T_x$ .
For this purpose  we complete $x$ to an orthonormal basis $\{e_v,e_a,e_b, e_1,\dots,e_{k-2}\}$ for $\R^{k+1}$. We define a basis $B_x$ for $T_x$ as $\{e_{va},e_{vb},e_{ab},e_{vi},e_{ai},e_{bi}|i=1,\dots,k-2\}$, where 
\begin{equation}\label{E:normalization}
e_{ij}=\frac{1}{\sqrt{2}}(e_ie_j-e_je_i)
\end{equation}
$i,j=v,a,b,1,\dots,k-3$. Indeed each  $e_{ij}\in B_x$ defined a variation of the frame :
\begin{equation}\label{E:variation}
x\rightarrow x+\epsilon \frac{1}{\sqrt{2}}((e_i\cdot e_v)e_j-(e_j\cdot e_v)e_i,(e_i\cdot e_a)e_j-(e_j\cdot e_a)e_i,(e_i\cdot e_b)e_j-(e_j\cdot e_b)e_i)
\end{equation} where $\epsilon$ is an infinitesimal parameter. The metric $h$ in this frame is
\begin{equation}\label{E:refmetric}
g_{ref}=de_{va}^2+de_{vb}^2+de_{vb}^2 +\sum_{i=1}^{k-2}de_{vi}^2+de_{ai}^2+de_{bi}^2
\end{equation}

In the above notations the pullback of the metric $g$ (\ref{E:flatmetric})  is equal to 
\[\begin{split}
& alg^* g=\frac{R^2}{4}\left(\frac{1}{t(1-t)}dt^2+ g_{\Omega}(t)\right), \quad trig^* g=d\tau^2+\frac{R^2}{4} g_{\Omega}(\sin^2\left(\tau/R\right))\\
&g_{\Omega}(t)=(1+t)de_{va}^2+(1+t)de_{vb}^2+
2(1-t)de_{ab}^2+\\
&+\sum_{i=1}^{k-2}2tde_{vi}^2+
(1-t)de_{ai}^2+(1-t)de_{bi}^2
\end{split}\]
\[\begin{split}
&trig^* g=\frac{R^2}{4}\left(\frac{4}{R^2}d\tau^2+
\left(1+\sin^2\left(\frac{\tau}{R}\right)\right)de_{va}^2+\left(1+\sin^2\left(\frac{\tau}{R}\right)\right)de_{vb}^2+
2\cos^2\left(\frac{\tau}{R}\right)de_{ab}^2+\right.\\
&\left.+\sum_{i=1}^{k-2}\left(2\sin^2\left(\frac{\tau}{R}\right)de_{vi}^2+
\cos^2\left(\frac{\tau}{R}\right)de_{ai}^2+\cos^2\left(\frac{\tau}{R}\right)de_{bi}^2\right) \right)\end{split}\]

\[
\begin{split}
&alg_t^*g=\frac{R^2}{4}\left(\frac{1}{t(1-t)}dt^2+
(1+t)de_{va}^2+(1+t)de_{vb}^2+
2(1-t)de_{ab}^2+\right.\\
&\left.+\sum_{i=1}^{k-2}2tde_{vi}^2+
(1-t)de_{ai}^2+(1-t)de_{bi}^2\right)\end{split}\]

\begin{proposition}\ \\
\\
\begin{enumerate}
\item  
\[\begin{split}
&s(t)=\sqrt{\det G_{\Omega}(t)}=\d vol_{g_{\Omega}}/ \d vol_{g_{ref}} ={2}^{\frac{k-2}{2}} t^{\frac{k-2}{2}}(1-t)^{k-\frac{1}{2}}(1+t)\\
&\sqrt{\det G_{h\Omega}(t)}=\d vol_{g_{h\Omega}}/ \d vol_{g_{ref}}={2}^{\frac{3k-2}{2}}t^{\frac{k-2}{2}}(1-t)^{\frac{2k-3}{2}}
\end{split}
\]
\item

 %Let $G'=dr^2+G$ be the product metric on $(0,R^2)\times V_3(\R^{k+1})$. 
Let $G=d\tau^2+g_{ref}$ be the product metric on $(0,\frac{\pi R}{2})\times V_3(\R^{k+1})$.  
The quantity $dvol_{trig^*g}/dvol_{G}$ is equal to 
\begin{equation}\label{E:trig}
w^k_{trig}(\tau)=\frac{R^{3k-3}}{2^{\frac{5k-5}{2}}}\sin^{k-2}\left(\tau/R\right)\cos^{2k-3}\left(\tau/R\right)(1+\sin^{2}\left(\tau/R\right))
\end{equation}
The volume of $(0,\frac{\pi R}{2})\times V_3(\R^{k+1})$  with respect to $dvol_{trig^*g}$ is equal to
\[\frac{2^{3-\frac{5 k-1}{2}} R^{3 k-2} \Gamma (k-1) \Gamma \left(\frac{k+1}{2}\right)}{\Gamma \left(\frac{3 k}{2}-\frac{1}{2}\right)}\times \frac{\pi}{2}\times Vol_{g_{ref}}(V_3(\R^{k+1}))\]
\item  Let $G'=dt^2+g_{ref}$ be the product metric on $(0,1)\times V_3(\R^{k+1})$. The quantity $dvol_{alg^*g}/dvol_{G}$ is equal to 

\begin{equation}\label{E:walg}
\begin{split}
&w^k_{alg}(t)=c_kt^{\frac{k-3}{2}}(1-t)^{k-2}(1+t) \text{ where }\\ %\frac{R^{3k-2}}{2^{\frac{5k-2}{2}}}(1+t)t^{\frac{k-3}{2}}(1-t)^{k-2}
&c_k=\frac{R^{3k-2}}{2^{\frac{5k-3}{2}}}
\end{split}
\end{equation}

\end{enumerate}
\end{proposition}

In the following we will omit explicit $k$-dependence of $w^k_{trig}$, $w^k_{alg}$ when the value of $k$ is clear from the context.

The volume of a unit $k$-sphere is equal to (see e.g. \cite{Huber}) to $Vol(S^k)=\frac{2\pi^{\frac{k+1}{2}}}{\Gamma(\frac{k+1}{2})}$ where $\Gamma(x)$ is the gamma function. Projection on the first vector defines a homogeneous submersion $V_3(\R^{k+1})\rightarrow S^k$. From this we deduce that \[Vol_{g_{ref}}(V_3(\R^{k+1}))=
\frac{8\pi^{\frac{3k}{2}}}{\Gamma(\frac{k+1}{2})\Gamma(\frac{k}{2})\Gamma(\frac{k-1}{2})}\]

\section{Ricci curvature}\label{S:ricci}
Such information could be useful because general theorems concerning Schr\"odinger operator (like \ref{T:M}) rely on it.
\subsection{Some explicit computations for $\L_1(S^k)$}
In this section we present results of computation of Ricci curvature of $\L_1(S^k)$ carried out with a help of {\it Mathematica}.  Here we present some partial but explicit results.

Ricci operator $Ric_i^j=Ric_{ik}g^{kj}$ be diagonalized in a $g_{ij}$-othonormal basis. 

For $k=2$ one can compute all relevant tensors explicitly using {\it RGTC Mathematica} pakage.  The eigen values for the space  $\L_1(S^2)$ are $\frac{3 t^2+6 t-1}{\text{R}^2 (t+1)^2}$ having multiplicity $3$ and $\frac{3 t^2+2 t+3}{\text{R}^2 (t+1)^2}$ with multiplicity one. In particular, the scalar curvature is $\frac{4 t (3 t+5)}{R^2(t+1)^2}$. From this we see that  for  $\L_1(S^2)$ \[Ric_{ij}\geq -\frac{1}{R^2}g_{ij}.\]

The map
\begin{equation}\label{E:projection}
p:\L_1(S^k)\rightarrow (0,R),(v,a,b)\rightarrow |v|/R
\end{equation}is a submersion because it commutes with the action of $\SO(k+1)$. It is defined on  $L_k\subset \L_1(S^k)$ which consists of $a,b,v$ such that $||a||\neq 0,R$.

It is natural to use O'Neill formulas \cite{ONeill} for computation of the  Ricci curvature of  $\L_1(S^k)$. The fibers of the map $p$ are Stiefel manifolds that carry $\SO(k+1)$-homogeneous metric.  Besse in \cite{Besse} gives formulas for curvature tensors of homogenous manifold that we are going to use. We have to set notations first.
Let $G$ be a Lie group and $F$ be a closed connected Lie subgroup. We denote by $\g$ and $\ff$ the Lie algebras of $G$ and $F$ respectively. We assume that $\g$ splits 
\begin{equation}\label{E:splits}
\g=\ff+\fp
\end{equation} into a direct sum of $\ff$ representations. We identify the tangent space to $eF\in G/F$ with $\g/\ff\cong \fp$. The $G$-invariant metric on $G/F$ defines an inner product on $\fp$ and is completely characterized by it.  All curvature tensors are $G$-invariant and can be written purely in terms of the inner product $(.,.)$ on $\fp$, bracket in $\g$ and decomposition (\ref{E:splits}). We denote by $[a,b]_{\ff}$  projection of $[a,b]$ onto $\ff$ and by $[a,b]_{\fp}$ the corresponding projection on $\fp$. According to \cite{Besse} the formula for Riemann tensor $R(X,Y)\in End(\fp)$ reads as
\[\begin{split}
&(R(X,Y)X,Y)=\frac{3}{4}||[X,Y]_{\fp}||^2-\frac{1}{2}([X,[X,Y]]_{\fp},Y)-\frac{1}{2}([Y,[Y,X]]_{\fp},Y)\\
&+||U(X,Y)||^2-(U(X,X),U(Y,Y)),\quad X,Y\in \fp
\end{split}\]
The map $U:\fp\otimes\fp\rightarrow \fp$ is defined by the formula:
\[2(U(X,Y),Z)=([Z,X]_{\fp},Y)+(X,[Z,Y]_{\fp})\]
In order to write a  formula for Ricci curvature $Ric$ we have to fix orthogonal basis $\{X_i\}$for $\fp$. Then
\[
\begin{split}
&Ric(X,X)=-\frac{1}{2}\sum_{j}||[X,X_j]_{\fp}||^2-\frac{1}{2}\sum_{j}([X,[X,X_j]_{\fp}]_{\fp},X_j)-\frac{1}{2}\sum_{j}([X,[X,X_j]_{\ff}]_{\fp},X_j)\\
&+\frac{1}{4}\sum_{i,j}([X_i,X_j]_{\fp},X)^2-([Z,X]_{\fp},X)
\end{split}\]
where $Z=\sum_{i}U(X_i,X_i)$. It is known that $(Z,X)=\tr(ad X)$. This is why $Z=0$, when the group $G$ is $\SO(k+1)$.

The formula for the Ricci tensor of the fiber of the projection $p$ is 
\[\begin{split}
&Ric^p=\frac{(2 (k+1)-4) t^2+(4 (k+1)-16) t+2 (k+1)-4}{(1+t)^2}dv_{ab}^2+\\
&\frac{(2 (k+1)-4) t+2 (k+1)-6}{t+1}\left(dv_{va}^2+dv_{vb}^2\right)\\
&+(k-3)\sum_{i=1}^{k-2} (dv_{vi}^2+dv_{ai}^2+dv_{bi}^2),k\geq3\\
&Ric^p=\frac{2 t^2-4 t+2}{(1+t)^2}dv_{ab}^2+\frac{2t}{t+1}\left(dv_{va}^2+dv_{vb}^2\right),\\ &k=2
\end{split}\]
$Ric^p$ is positive definite when $n\geq 3, t\geq 0$

In order to compute second fundamental form of the fibers of the projection $p$ we choose a vector field $e'=\frac{\sd}{\sd t}$ which is orthogonal to the fibers of $p$.
The formula for the form reads
\[T(\eta,\xi)=\frac{1}{|e'|}\nabla_{\eta}\xi\cdot e'\] where $\nabla$ is the Levi-Civita connection. We can compute the same quantity If we extend $\eta,\xi,e'$ from $\L_1(S^k)$ to $\Trig_1(\R^{k+1})$ compute it there and then restrict back to $\L_1(S^k)$. For connection we take  Levi-Civita connection on $\Trig_1(\R^{k+1})$.

If vector fields $\eta,\xi$ are elements of $\so_{k+1}$, then they have a canonical extension to $\Trig_1(\R^{k+1})$, produced by $\so_{k+1}$ action. The vector field $e'$ is a restriction of the vector field
\[\frac{R^2}{2|v|^2}v\cdot \sd_v-\frac{R^2}{2(R^2-|v|^2)}a\cdot \sd_a-\frac{R^2}{2(R^2-|v|^2)}b\cdot \sd_b\]
which is defined on $\Trig_1(\R^{k+1})$, for which we will keep the same notation. Components $v\cdot \sd_v, a\cdot \sd_a,b\cdot \sd_b$ are dilation operators acting on separately on $v,a$ and $b$ components. Advantage of working on $\Trig_1(\R^{k+1})$ is that covariant derivative associated with Levi-Civita connection defined by partial derivatives.
In order to simplify result of coordinate computation  it is convenient to evaluate $T$ at a point with $v=(Rt^{1/2},0,0,\dots,0),a=(0,R\sqrt{1-t},0,\dots,0),b=(0,0,R\sqrt{1-t},0,\dots,0)$.
The result of the computation is
\[T=\frac{R}{2}\sqrt{t(1-t)}\left(dv_{va}^2 +dv_{vb}^2-2dv_{ab}^2+\sum_{i=1}^{k-2} (2dv_{vi}^2-dv_{ai}^2-dv_{bi}^2)\right)k\geq 3\]
\[T=\frac{R}{2}\sqrt{t(1-t)}(dv_{va}^2 +dv_{vb}^2-2dv_{ab}^2),\quad k=2\]
The trace of $T$ is equal to 
\[\begin{split}
&\tr T=
\frac{2 \left((7-4 k) t^2+(7-3 k) t+k-2\right)}{R  (t+1)\sqrt{t(1-t) }},\quad k\geq3 \\
&\tr T= 4\frac{\sqrt{t(1-t)}}{R (t+1)},\quad k=2
\end{split}\]
The tensor $CT=T_{ik}g^{kk'}T_{kj}$ is equal to
\[
\begin{split}
&CT=t(1-t)\left( \frac{1}{1+t}dv_{va}^2 +\frac{1}{1+t}dv_{vb}^2+\frac{2}{1-t}dv_{ab}^2+\sum_{i=1}^{k-2} \frac{2}{t}dv_{vi}^2+\frac{1}{1-t}dv_{ai}^2+\frac{1}{1-t}dv_{bi}^2\right) k\geq 3\\
&CT=t(1-t)\left( \frac{1}{1+t}dv_{va}^2 +\frac{1}{1+t}dv_{vb}^2+\frac{2}{1-t}dv_{ab}^2\right) k= 2
\end{split}
\]
Finally the vertical component of the Ricci tensor is equal to 
\[\begin{split}
&-\frac{\left((8 k-7) t^2+(5 k-6) t-3 k+5\right) \left({de}_{ai}^2+\text{de}_{bi}^2\right)}{t+1}\\
&+\frac{\left(16 k t^2+13 k t-3 k-14 t^2-19 t+3\right) {de}_{vi}^2}{t+1}\\
&-\frac{2 {de}_{ab}^2 \left((8 k-7) t^3+(13 k-15) t^2+(2 k-1) t-3 k+3\right)}{(t+1)^2}+((8 k-7) t-2) \left({de}_{{va}}^2+{de}_{{vb}}^2\right)
\end{split}\]

We will not attempt to compute the horizontal and mixed components here and just notice that if $k\geq 4$ then the mixed component must be zero by symmetry reasons. The general for the Ricci curvature for arbitrary $\L_N(S^k)$ will be presented in the next section. 
\subsection{The curvature of a submanifold $M\subset \mathbb{R}^n$. A review}
The space $\L_N(S^k)$ (\ref{E:LNdef}) are defined as a submanifold inside f the  linear space $\Trig_N(\R^{k+1})$. Curvature of such submanifold can be computed from fundamental equations of differential geometry.  In this section we are going to remind how this is done and set up notations.  Let $(f_{1}(v),\dots,f_{k}(v))$ be components of a a smooth map $f:\mathbb{R}^n\rightarrow \mathbb{R}^k$. Let us assume that $f$ is transversal at $0\in \mathbb{R}^k$, that is  collection of smooth  $df_{1},\dots,df_{k}$ are linearly independent on  
\begin{equation}\label{E:Man}
M=\{v\in \mathbb{R}^n|f(v)=0\}.
\end{equation}  By inverse function theorem $M$ is a smooth submanifold in $\mathbb{R}^n$. Introduce notation: $T_M$ and $N_M$ stands for the tangent and normal bundles of $M\subset \mathbb{R}^n$, $\Gamma(M,L)$ is a space of $C^{\infty}$-sections of a vector bundle $L$ over $M$.
Let $R'$ be the curvature of a   metric tensor $g'_{ab}$ on $\mathbb{R}^n$. The curvature $R$ of the induced metric $g_{cd}$ on $M$ can be computed via Gauss'-Weingarten equation
\[\langle R'(X,Y)Z, W\rangle = \langle R(X,Y)Z, W\rangle + \langle \alpha(X,Z), \alpha(Y,W)\rangle -\langle \alpha(Y,Z), \alpha(X,W)\rangle \]
where $\alpha\in \Gamma(M,\Sym^2T_M\otimes N_M)$ is the second quadratic form. We can compute it by the formula 
\[\alpha(X,Y)=-\langle X,\nabla_Y e_i \rangle e_i.\]
Here $e_i$ is an orthonormal basis for $N_M$. Obviously we get the same answer if we use the formula
\begin{equation}\label{E:secondfundamental}
\alpha(X,Y)=\langle \nabla_XY,h_i \rangle h^i
\end{equation}
where 
\begin{equation}\label{dualbases}
\{ h_i |i=1,\dots,k\}\text{ is a basis for }N_M \text{ and $\{ h^i \}$ is the $g'$-dual basis}.
\end{equation}
 In the following indices labelled by latin letters $i,j,k,l,q,r$ will run through this range. Then
\[ \langle \alpha(X,Z), \alpha(Y,W)\rangle =\langle \nabla_XZ,h_i \rangle \langle \nabla_Y W, h_k \rangle  s^{ik}\]
\[s_{ik}= \langle h_i,h_k \rangle\]
and $(s_{ik})$ is the inverse to $(s^{ik})$.
 In our case collection 
 \begin{equation}\label{E:basis}
 \{ h_i\}=\{ \mathrm{grad}\ f_i\}
 \end{equation}
  defines a local  basis for $N_M$. Furthermore 
 \begin{equation}\label{E:gij}
 \text{ we assume that $g'_{ab}$ written in standard coordinates is a constant matrix.}
 \end{equation}
  Then the covariant derivative $\nabla'_a$ associated with Levi-Civita connection is $\frac{\sd}{\sd x^a}$, $R'=0$ and
 \begin{equation}\label{E:Rcurv}\langle R(X,Y)Z, W\rangle=
 \langle \nabla'_YZ, \mathrm{grad}f_i\rangle \langle \nabla'_XW, \mathrm{grad}f_j\rangle s^{ij}
-\langle \nabla'_XZ, \mathrm{grad}f_i\rangle \langle \nabla'_YW, \mathrm{grad}f_j\rangle s^{ij}
  \end{equation}
 \[s_{ij}=\langle \mathrm{grad}f_i,\mathrm{grad}f_j \rangle_{g'}\]
 We just proved the following proposition
 \begin{proposition}
 The curvature of  the submanifold $M\subset \mathbb{R}^n$ (\ref{E:Man}), where $\mathbb{R}^n$ is equipped with a metric $g'_{ab}$ (\ref{E:gij}) is given by the formula (\ref{E:Rcurv}).
 \end{proposition}
For $M$-tangential vectors $\langle Z, \mathrm{grad}f_i\rangle=0$. Then $\langle \nabla'_YZ, \mathrm{grad}f_i\rangle=-\langle Z,  \nabla'_Y\mathrm{grad}f_i\rangle$ and the formula can be written in a form that is more suitable for computations:
 
  \begin{equation}\label{E:Rcurvf}
  \begin{split}
&  \langle R(X,Y)Z, W\rangle=\\
&=\langle \nabla'_Y Z, \mathrm{grad}f_i\rangle \langle \nabla'_XW, \mathrm{grad}f_j\rangle s^{ij}
-\langle \nabla'_XZ, \mathrm{grad}f_i\rangle \langle \nabla'_WY, \mathrm{grad}f_j\rangle s^{ij}
\end{split}
  \end{equation}
  
Let $\{l_i\}$ be a local basis of the tangent bundle $T_M$, and $\{l^i\}$ be the $g'$-dual basis. The tensor $g'^{ab}\frac{\sd}{\sd x^a}\otimes \frac{\sd}{\sd x^b}, 1\leq ab\leq n$ is  a sum  of two orthogonal components  $l_i\otimes l^i+h_r\otimes h^r$. Then 
\begin{equation}\label{E:orthdec}
l_i\otimes l^i=g'^{ab}\frac{\sd}{\sd x^a}\otimes\frac{\sd}{\sd x^b}-h_r\otimes h^r.
\end{equation}
Recall that the Ricci tensor is a contaction
\[Ric(X,Y)= \langle R(X,l_i)Y,l^{i}\rangle=- \langle R(X,l_i)l^{i},Y\rangle\]

 We use the second sum in the above formula and (\ref{E:orthdec}) to derive $Ric(X,W)$:

 \[\begin{split}&Ric(X,W)=
g'^{ab} \langle \nabla'_{\frac{\sd}{\sd x^a}}\frac{\sd}{\sd x^b}, \mathrm{grad}f_i\rangle \langle \nabla'_XW, \mathrm{grad}f_j\rangle s^{ij}
-g'^{ab}\langle \nabla'_X\frac{\sd}{\sd x^a}, \mathrm{grad}f_i\rangle \langle \nabla'_W\frac{\sd}{\sd x^b}, \mathrm{grad}f_j\rangle s^{ij}\\
&-\langle \nabla'_{h_r}h^r, \mathrm{grad}f_i\rangle \langle \nabla'_XW, \mathrm{grad}f_j\rangle s^{ij}
+\langle \nabla'_X h^r, \mathrm{grad}f_i\rangle \langle \nabla'_W{h_r}, \mathrm{grad}f_j\rangle s^{ij}
\end{split}  \]
We know that $\frac{\sd}{\sd x^a}$ is covariantly constant so $\nabla'_X\frac{\sd}{\sd x^a}=0$ for all $X$. Additionally \[h^q=s^{qr}\mathrm{grad}f_r\] and
\begin{equation}\label{E:simp}
\langle \nabla'_{X}h^r,\mathrm{grad}f_i\rangle=-\langle h^r, \nabla'_{X}\mathrm{grad}f_i\rangle\text{ because } \langle h^r,\mathrm{grad}f_i\rangle=\delta_{i}^r
\end{equation}  Thus 
 \[\begin{split}&Ric(X,W)=
 -\langle \nabla'_{\mathrm{grad}f_r}h^r, \mathrm{grad}f_i\rangle \langle \nabla'_XW, \mathrm{grad}f_j\rangle s^{ij}
+\langle \nabla'_X h^r, \mathrm{grad}f_i\rangle \langle \nabla'_{W}\mathrm{grad}f_r, \mathrm{grad}f_j\rangle s^{ij}=\\
& \langle \mathrm{grad}f_q, \nabla'_{\mathrm{grad}f_r}\mathrm{grad}f_i\rangle \langle \nabla'_X W,   \mathrm{grad}f_j\rangle s^{qr} s^{ij}
-\langle \mathrm{grad}f_q, \nabla'_X\mathrm{grad}f_i\rangle \langle \mathrm{grad}f_r, \nabla'_{W}\mathrm{grad}f_j\rangle s^{qr}s^{ij}
\end{split}  \]
\begin{proposition}
Ricci curvature of  the submanifold $M\subset \mathbb{R}^n$ (\ref{E:Man}), where $\mathbb{R}^n$ is equipped with a metric $g'_{ij}$ (\ref{E:gij}) is equal to 
\begin{equation}\label{E:riccigen}
\begin{split}&Ric(X,W)= \langle \mathrm{grad}f_q, \nabla'_{\mathrm{grad}f_r}\mathrm{grad}f_i\rangle \langle \nabla'_XW,   \mathrm{grad}f_j\rangle s^{qr} s^{ij}\\
&-\langle \mathrm{grad}f_q, \nabla'_X\mathrm{grad}f_i\rangle \langle \mathrm{grad}f_r, \nabla'_{W}\mathrm{grad}f_j\rangle s^{qr}s^{ij}\end{split}\end{equation}
where $X,W\in \Gamma(M,T_M)$.
\end{proposition}
\begin{corollary}
The scalar curvature $Sc$ of $M$ is equal to 
\begin{equation}\label{E:Sc}
\begin{split}
&Sc=\\
&-\langle \mathrm{grad}f_q,\mathrm{grad}\frac{\sd f_i}{\sd x^a}\rangle\langle \mathrm{grad}f_r,\mathrm{grad}\frac{\sd f_j}{\sd x^b}\rangle g'^{ab}s^{qr} s^{ij}\\
&+\langle \mathrm{grad}f_q,\nabla'_{\mathrm{grad}f_r}\mathrm{grad}f_i\rangle \langle \mathrm{grad}f_k,\nabla'_{\mathrm{grad}f_l}\mathrm{grad}f_j\rangle s^{kl} s^{qr} s^{ij}\\
&+\langle \mathrm{grad}f_q,\nabla'_{\mathrm{grad}f_k}\mathrm{grad}f_i\rangle \langle \mathrm{grad}f_l,\nabla'_{\mathrm{grad}f_r}\mathrm{grad}f_j\rangle s^{kl} s^{qr} s^{ij}
\end{split}\end{equation}
\end{corollary}
\begin{proof}
Follows from the formula $R=R(l_i,l^i)$, (\ref{E:orthdec}) and (\ref{E:simp})
\end{proof}
The mean curvature $H$ is equal to 
\[\alpha(l_i,l^i)=\langle \nabla_{l_j}l^j,h_i \rangle h^i=g^{ab}\langle \nabla_{\frac{\sd}{\sd x^a}}\frac{\sd}{\sd x^b},h_i \rangle h^i-\langle \nabla_{h_r}h^r,h_i \rangle h^i=-\langle \nabla_{h_r}h^r,h_i \rangle h^i=\] 
\[=\langle\mathrm{grad}f_q, \nabla_{\mathrm{grad}f_r}\mathrm{grad}f_i\rangle\mathrm{grad}f_js^{rq}s^{ij}\]
It square is equal to 
\begin{equation}\label{E:meanssquare}
H^2=s^{rq} \langle \mathrm{grad}f_q,\nabla_{\mathrm{grad}f_r}\mathrm{grad}f_i\rangle s^{kl}\langle \mathrm{grad}f_{k},\nabla_{\mathrm{grad}f_{l}}\mathrm{grad}f_{j}\rangle s^{ij}
\end{equation}
Note that it is the middle term in the formula for $Sc$.
\begin{proposition}\cite{Leung}
 Let $M^n$ be an immersed submanifold in the Euclidean space $\R^{n+p}$ where $p$ denotes the codimension. Let $Ric_{\min}$, $Sc$, and $H$ denote the functions that assign to each point of $M$ the minimum Ricci curvature, the scalar curvature, and the mean curvature respectively of M at the point. Then we have
 \begin{equation}\label{E:ricineq}
Ric_{\min} \geq  Sc - \frac{(n-1)H^2}{4}+ \frac{1}{4n^2} \left(\sqrt{n-1}(n-2)|H|-2   \sqrt{(n-1)H^2-nSc}\right)^{2}
\end{equation}
\end{proposition}

\subsection{Curvature tensors of $\L_N(S^k)$}
So far discussion was very general. Suppose now $\mathbb{R}^{n}\cong \Trig_N(\R^{k+1})$. It carries a metric
metric $\langle \delta n_1(\theta), \delta n_2(\theta)\rangle=\frac{1}{2\pi}\int_0^{2\pi} \delta n_1(\theta)\cdot \delta n_2(\theta)d\theta$, which in our coordinates  is 
\begin{equation}\label{E:flatmetricsincos}
dv\cdot dv+\frac{1}{2}\sum_{s=1}^N\left(da[s]\cdot da[s]+db[s]\cdot db[s]\right)
\end{equation}

A trigonometric polynomial $\phi(\theta)\in C_{2N}(S^1)$  defines a quadratic functional   
\begin{equation}\label{E:qfunct}
n(\theta)\rightarrow \frac{1}{2\pi}\int_{S^1}(n(\theta)\cdot n(\theta) -R^2)\phi(\theta)d\theta=f_{\phi}(n)
\end{equation} on $\Trig_N(\R^{k+1})$.
We defined  the real algebraic variety $\L_N(S^k)\subset \Trig_N(\R^{k+1})$ as  $\{n(\theta)\in\Trig_N(\R^{k+1})| f_{\phi}(n(\theta))=0\ \forall \phi\in  C_{2N}(S^1)\}$. 
 
It will be useful to rewrite   metric (\ref{E:flatmetricsincos})  by using  Dirichlet kernel (\ref{E:dirichlet3})
 \begin{equation}\label{E:innerexpDirishlet}
 \begin{split}
 &\langle \delta_1 n, \delta_2 n \rangle
 =\frac{1}{4\pi^2}\int_{S^1\times S^1} \delta_1 n(\theta_1)\cdot \delta_2 n(\theta_1)D_N(\theta_1-\theta_2)
 \end{split}\end{equation}
 
 The tangent bundle to $\R^{k+1}$ is canonically trivialized. This is why we can identify tangent vectors $T_{n(\theta)}(\Trig_N(\R^{k+1}))$ with elements $e(\theta)\in \Trig_N(\R^{k+1})$.

We start curvature computations  with writing down the formula for  $\mathrm{grad} f_{\phi}$.
\begin{proposition}\label{P:gradgen}
\begin{enumerate}
\item  
\begin{equation}\label{E:vectn}
\mathrm{grad} f_{\phi}(\theta)=\frac{1}{\pi}\int_{S^1}D_N(\theta-\theta')n(\theta')\phi(\theta')\rd \theta'
\end{equation}
\item Consider  the bilinear form \[4g_N(\phi,\psi):=\langle \mathrm{grad} f_{\phi},\mathrm{grad} f_{\psi} \rangle=\frac{1}{4\pi^2}\int_{S^1\times S^1}4g_N(\theta,\theta')\phi(\theta)\psi(\theta')\rd \theta \rd \theta'\] on $C_{2N}(S^1)$.
The kernel $g_N$  is equal to 
\begin{equation}\label{E:kern}
g_N(\theta,\theta')=D_N(\theta-\theta')n(\theta)\cdot n(\theta')
\end{equation}
\end{enumerate}
\end{proposition}
\begin{proof}
\begin{enumerate}
\item Let $n_t(\theta)$ be a sooth path $t\in [0,\epsilon)$, $n_0(\theta)=n(\theta)$ Follows from  the formula 
\[\left. \frac{\sd f_{\phi}}{\sd t}\right|_{t=0}=\frac{2}{2\pi}\int_{S^1}\left.n(\theta)\cdot \frac{\sd n_{t}(\theta)}{\sd t}\right|_{t=0} \phi(\theta)\rd \theta
=\frac{1}{\pi}\int_{S^1\times S^1}\left.n(\theta')\cdot \frac{\sd n_{t}(\theta)}{\sd t}\right|_{t=0}\phi(\theta')D_N(\theta-\theta')\rd \theta\rd \theta'
\]
\item 
\[\begin{split}&4g(\phi,\psi)=\frac{4}{8\pi^3}\int_{S^1\times S^1\times S^1}D_N(\theta-\theta')D_N(\theta-\theta'')n(\theta')\cdot n(\theta'')\phi(\theta')\phi(\theta'')\rd \theta\rd \theta'\rd \theta''=\\
&=\frac{4}{4\pi^2}\int_{S^1\times S^1}D_N(\theta'-\theta'')n(\theta')\cdot n(\theta'')\phi(\theta')\phi(\theta'')\rd \theta'\rd \theta''\end{split}\]
We used 
\begin{equation}\label{E:idempotent}
\frac{1}{2\pi}\int_{S^1}D_N(\theta-\theta')D_N(\theta-\theta'')\rd \theta=D_N(\theta'-\theta'')
\end{equation}
 %\[\frac{1}{2\pi \sqrt{-1}}\int_{|u|=1}(2S_N(u/s)n(s))\cdot (2S_N(u/t)n(t))\frac{du}{u}=4S_N(t/s)n(s)\cdot n(t)\]
\end{enumerate}

\end{proof}

Besides inner product $g(\cdot,\cdot)$ the space $C_{2N}(S^1)$ carries the standard $L^2$-inner product $(\phi,\psi)=\frac{1}{2\pi}\int_{S^1}\phi(\theta)\psi(\theta)\rd (\theta)$. Then $g(\phi,\psi)=(g(\phi),\psi)$, where $g$ is some symmetric operator on $C_{2N}(S^1)$. The inverse by 
\begin{equation}\label{E:Gdef}
G(\psi)(\theta)=\frac{1}{2\pi}\int_{S^1}G(\theta,\theta')\psi(\theta')\rd\theta'
\end{equation} has  the kernel that satisfies
\[\frac{1}{2\pi}\int_{S^1}D_N(\theta-\theta')n(\theta)\cdot n(\theta') G(\theta',\theta'')\rd \theta'=D_{2N}(\theta-\theta'')\]

Let $\nabla$ be the covariant derivative associated with $\langle\cdot,\cdot \rangle$-compatible  Levi-Civita connection on $T\Trig_l(\R^{k+1})$.
\begin{proposition}\label{P:covdev}
Let $e(\theta)$ be a tangent vector to $\Trig_l(\R^{k+1})$.  
Then \[(\nabla_{\e}\mathrm{grad} f_{\psi})(\theta)=\frac{1}{2\pi}\int_{S^1}2D_N(\theta-\theta')e(\theta')\psi(\theta')\rd \theta'\]
\end{proposition}

\begin{proof} 
 The action of the operator $\nabla_{\e}$ on tensor fields has a very simple description in our coordinates. It acts as $\frac{\sd}{\sd e}$ component-wise.The vector field $\mathrm{grad} f_{\phi})$ has linear coefficients. Introduce temporally coordinates  $\{n^{\alpha}\}$ $T\Trig_l(\R^{k+1})$ in which metric has the identity Gram matrix.  $\nabla_{\e}\mathrm{grad} f_{\phi}$ is obtained by replacing all occurrences of $n^{\alpha}$ by $e^{\alpha}$. Then the  formula immediately follows from (\ref{E:vectn}).
\end{proof}
\begin{corollary}
The vector field 
$(\nabla_{\mathrm{grad} f_{\phi}}\mathrm{grad} f_{\psi})(\theta)$ is equal 
\[ \frac{1}{4\pi^2}\int_{S^1\times S^1}4D_N(\theta-\theta')D_N(\theta'-\theta'')n(\theta'')\phi(\theta'')\psi(\theta')\rd \theta'\rd \theta''\]
\end{corollary}
The formula for Ricci curvature (\ref{E:riccigen}) involves contraction  of several tensors. 
In the following corollary we list results of a computation of various   components of Ricci tensor:

\begin{corollary}\label{C:innerprodcomp}\ \\
\noindent Fix in addition two vector fields  are $X(\theta),W(\theta)$. Then 
\begin{enumerate}
\item The kernel of tri-linear form \[\langle \mathrm{grad} f_{\eta}, \nabla_{\mathrm{grad} f_{\phi}}\mathrm{grad} f_{\psi}\rangle=\frac{1}{8\pi^3}\int_{S^1\times S^1\times S^1}K(\theta,\theta',\theta'')\eta(\theta)\phi(\theta')\psi(\theta'')\rd \theta \rd \theta' \rd \theta''\]
is equal to 
\[K(\theta,\theta',\theta'')=8D_N(\theta-\theta'')D_N(\theta'-\theta'')n(\theta)\cdot n(\theta')\]
\item For a fixed $W$ and $X$ the functional \[\frac{1}{2\pi}\int_{S^1}T_{X,T}(\theta)\phi(\theta)\rd \theta= \langle W, \nabla_{X}\mathrm{grad} f_{\phi}\rangle\] has the kernel 
\begin{equation}\label{E:onekern}
T_{X,T}(\theta) =2W(\theta)\cdot X(\theta)
\end{equation}
\item  For a fixed $X$ the kernel of the bilinear form $\frac{1}{4\pi^2}\int_{S^1\times S^1}R_X(\theta,\theta')\phi(\theta)\psi(\theta')\rd\theta \rd\theta'= \langle \mathrm{grad} f_{\phi}, \nabla_{X}\mathrm{grad} f_{\psi}\rangle$
is
\[R_X(\theta,\theta')=4D_N(\theta-\theta')n(\theta)\cdot X(\theta')\]
\item  For a fixed $W$ the kernel of the bilinear form $\frac{1}{4\pi^2}\int_{S^1\times S^1}S_W(\theta,\theta')\phi(\theta)\psi(\theta')\rd\theta \rd\theta'= \langle W, \nabla_{\mathrm{grad} f_{\phi}}\mathrm{grad} f_{\psi}\rangle$
is

\[S_X(\theta,\theta')=4D_N(\theta-\theta')n(\theta)\cdot W(\theta')\]
\end{enumerate}

\end{corollary}

\begin{proof}
\begin{enumerate}
\item \[\langle \mathrm{grad} f_{\eta}, \nabla_{\mathrm{grad} f_{\phi}}\mathrm{grad} f_{\psi}\rangle=\]\[=\langle \frac{1}{2\pi}\int_{S^1}2D_N(\mu-\theta)n(\theta)\eta(\theta)\rd \theta, \frac{1}{4\pi^2}\int_{S^1\times S^1}4D_N(\mu-\theta')D_N(\theta'-\theta'')n(\theta'')\phi(\theta'')\psi(\theta')\rd \theta'\rd \theta''\rangle, \mu\in S^1\]
\[=\frac{1}{8\pi^3}\int_{S^1\times S^1\times S^1}8D_N(\theta-\theta'')D_N(\theta'-\theta'')n(\theta)\cdot n(\theta')\eta(\theta)\phi(\theta')\psi(\theta'')\rd \theta \rd \theta' \rd \theta''\]
We used (\ref{E:idempotent}).
\item  \[\langle W, \nabla_{X}\mathrm{grad} f_{\phi}\rangle=\frac{1}{4\pi^2}\int_{S^1\times S^1}2W(\theta)\cdot X(\theta')D_N(\theta-\theta')\phi(\theta')\rd \theta \rd \theta'\]
\[=\frac{1}{2\pi}\int_{S^1}2W(\theta)\cdot X(\theta)\phi(\theta)\rd \theta\]
\item
\[\langle \mathrm{grad} f_{\phi}, \nabla_{X}\mathrm{grad} f_{\psi}\rangle=\frac{1}{8\pi^3}\int_{S^1\times S^1\times S^1}4[D_N(\mu-\theta)n(\theta)\phi(\theta)]\cdot [D_N(\mu-\theta')X(\theta')\psi(\theta')]\rd \mu\rd \theta\rd \theta'=\]
\[=\frac{1}{4\pi^2}\int_{S^1\times S^1}4D_N(\theta-\theta')n(\theta)\cdot X(\theta') \phi(\theta) \psi(\theta')\rd \theta\rd \theta'\] We used (\ref{E:idempotent}).
\item 
\[\langle W,  \nabla_{\mathrm{grad} f_{\phi}}\mathrm{grad} f_{\psi}\rangle=\frac{1}{8\pi^3}\int_{S^1\times S^1\times S^1}4W(\mu)\cdot [D_N(\mu-\theta')D_N(\theta-\theta')n(\theta)\phi(\theta)\psi(\theta')]\rd \mu\rd \theta\rd \theta'=\]
\[=\frac{1}{4\pi^2}\int_{S^1\times S^1}4n(\theta)\cdot W(\theta') D_N(\theta-\theta') \phi(\theta) \psi(\theta')\rd \theta\rd \theta'\]
\end{enumerate}
\end{proof}

An immediate corollary of the formulas (\ref{E:Rcurvf},\ref{E:onekern}) and of the definition of $G$ (\ref{E:Gdef}) is that the full curvature tensor is equal to 

\[\langle R_N(X,Y)Z, W\rangle=\frac{1}{4\pi^2}\int_{S^1\times S^1} \left(Z(\theta)\cdot Y(\theta) X(\theta')\cdot W(\theta')-Z(\theta)\cdot X(\theta) W(\theta')\cdot Y(\theta')\right)G_N(\theta,\theta')\rd \theta \rd\theta'\]

We can put together results of the computations and obtain from equation (\ref{E:riccigen}) and Corollary \ref{C:innerprodcomp} a formula for the Ricci curvature:
\[\begin{split}
&Ric(W,X)=\frac{1}{16\pi^4}\int_{S^1\times S^1\times S^1\times S^1}\left(D_N(\mu-\theta)G_N(\theta,\theta')n(\theta)\cdot n(\theta')D_N(\mu'-\theta')G_N(\mu,\mu') W(\mu)\cdot X(\mu)\right.\\
&\left.- D_N(\theta-\theta')n(\theta)\cdot X(\theta')  D_N(\mu-\mu')n(\mu)\cdot W(\mu')G_N(\theta,\mu)G_N(\theta',\mu')\right)\rd \theta \rd \theta' \rd \mu \rd \mu'
\end{split}\]

In order to use inequality (\ref{E:ricineq}) we compute the square of mean curvature $H^2$ and the scalar curvature. For computation of $H^2$ we use (\ref{E:meanssquare}) and (\ref{E:onekern}):
\begin{equation}\label{E:H2}
\begin{split}
&H^2=\\
&=\frac{1}{(2\pi)^6}\int_{(S^1)^{\times 6}} G(\theta,\theta') n(\theta)\cdot n(\theta') D_N(\theta-\theta'')D_N(\theta'-\theta'')G(\theta'',\mu'')\\
&D_N(\mu''-\mu)D_N(\mu''-\mu') n(\mu)\cdot n(\mu')G(\mu',\mu)\rd \theta\cdots \rd \mu''
\end{split}
\end{equation}

We would like contract $Ric(X,W)$ further and find scalar curvature.

\begin{proposition}
\begin{enumerate}
\[\begin{split}
&Sc=H^2+\\
&+\frac{1}{(2\pi)^6}\int_{(S^1)^{\times 6}}  n(\theta)\cdot n(\theta') D_N(\theta-\theta'')D_N(\theta'-\theta'')\\
&D_N(\mu''-\mu)D_N(\mu''-\mu') n(\mu)\cdot n(\mu')G(\theta,\mu)G(\theta',\mu')G(\theta'',\mu'')\rd \theta\cdots \rd \mu''\\
&-\frac{1}{(2\pi)^4}\int_{(S^1)^{\times 4}} D_N(\theta-\mu)G_N(\theta,\mu) n(\theta')\cdot n(\mu')G_N(\theta',\mu')D_N(\theta-\theta')D_N(\mu-\mu')\rd \theta\cdots \rd \mu'
\end{split}\]

\end{enumerate}
\end{proposition}
\begin{proof}
This is a straightforward adaptation of formula (\ref{E:Sc}) to the case of $L_N(S^k)$
\end{proof}

On manifolds $L_N(S^k)$ there are several ways to gauge the closeness of a point $n$ to a singular locus $L_{N-1}(S^k)$. Besides the most obvious way to do it with the distance function $\rho(n, L_{N-1}(S^k))$ we can use the function $n(\theta)\rightarrow f_N(n)=A_N^2$. The square of the last Fourier coefficients $ A_N^2= B_N^2$ have the same zero locus as $\rho(n, L_{N-1}(S^k))$ and algebraically more simple. This is why we use it in the definition of tubular neighborhood
\[U_{\epsilon}L_{N-1}(S^k)=\{n(\theta)\in L_N(S^k)|f(n)\leq \epsilon\}\]
\begin{proposition}
The set $U_{\epsilon}L_{N-1}(S^k)$ has a convex boundary $\{n(\theta)\in L_N(S^k)|f(n)=\epsilon\}$.

\end{proposition}
\begin{proof}
The second fundamental form (\ref{E:secondfundamental}) is the Hessian of $A_N^2$, evaluated on two tangent vectors $X,Y\in T_n(L_N(S^k))$. We used that connection $\nabla$ on $T(L_N(S^k))$ is induced from the  trivial connection on $\Trig_N(\R^{k+1})$.
Then $\alpha(X,Y)=2X_N\cdot Y_Nh$,  where $h=\frac{\mathrm{grad} f_{N}}{\mathrm{grad} f_{N}^2}$ is a vector normal to $\sd U_{\epsilon}L_{N-1}(S^k)$ such that $\langle \mathrm{grad},h\rangle=1$
The form $\alpha(X,Y)$ is semidefinite and $\sd U_{\epsilon}L_{N-1}(S^k)$ is convex.
\end{proof}
\section{Schr\"{o}dinger operator on $\L_1(S^k)$}\label{S:Schr}
Our goal is to define the Schr\"{o}dinger operator on $L^2(\L_1(S^k)^{smooth})$. It is the sum $-\Delta+U$, where $U$ is a restriction of the potential defined on $\Trig_1(\R^{k+1})$. Keep in mind that $U$ has  an intrinsic meaning for $L_1(S^k)$. The vector field $\rho$ is tangential to $M_1(S^k)\subset \Trig_1(\R^{k+1})$ and $pol^*g(\rho,\rho)=(\rho,\rho)=U$. Equations for $L_1(S^k)$ imply that \[U=\frac{1}{2L^2}(a^2+b^2)=\frac{a^2}{L^2}=\frac{R^2}{L^2}\cos^2\left(\frac{\tau}{R}\right)=\frac{R^2}{L^2}(1-t)\]

We use a densely defined quadratic form 
\begin{equation}\label{E:qshrodinger}
Q(f,g)=\int_{M_1(S^k)^{smooth}}(\nabla f\cdot \nabla \bar g+Uf\bar g)\d vol
\end{equation}
on the $L^2(M_1(S^k)^{smooth})\cap C_c^{\infty}(M_1(S^k)^{smooth})$ to 
define Schr\"{o}dinger operator $H$ as $Q(f,g)=(Hf,g)$ with 
\begin{equation}\label{E:selfadjoint}
(f,g)=\int_{M_1(S^k)^{smooth}}f \bar g \d vol
\end{equation}
The fibers of projection (\ref{E:projection}) are Stiefel manifolds we will start with analysis of the angular part of the Schr\"odinger operator that acts on functions on the fibers.
\subsection{Harmonic analysis on Stiefel manifolds}
We let $V_m(\R^{k+1})$ denote the Stiefel manifold 
of real matrices $V\in M_{k+1,m}$ such that
$ V^tV = I_m$. The rotation group $\SO(k+1)$ acts on 
$V_m(\R^{k+1})$ by left matrix multiplication so that $V_m(\R^{k+1})$ is isomorphic to $\SO(k+1)/\SO(k+1 - m)$. We would like to think about  $L^2(V_m(\R^{k+1})$ as  $L^2(\SO(k+1))^{\SO(k+1 - m)}$. The later space has an inner product $\langle f,g\rangle=\int_{\SO(k+1)}f\bar g \d\mu$, where $\d\mu$ is a normalized bi-invariant Haar measure on $\SO(k+1)$
According to \cite{Gelbart} if $k>2m$
\[L^2(V_m(\R^{k+1}))=\bigoplus_{\omega} H^{k+1,m}_{\omega}\otimes G_{\omega}\]
In this sum $H^{k+1,m}_{\omega}$ is an irreducible representation of $\SO(k+1)$ of highest weight $\omega=[m_1,\dots,m_{\left[\frac{k+1}{2}\right]}]$, with $m_1\geq \cdots\geq m_{\left[\frac{k+1}{2}\right]}$. The linear space $G_{\omega}$ coincides with the linear space of some  irreducible representation of $\GL(m)$ if $m_i=0$ for $i>m$. Otherwise $G_{\omega}=0$.

In our application we are interested in the case $m=3$.  Then \[\dim G_{[m_1,m_2,m_3]}=1/2 ( m_1-m_2 + 1) (m_2-m_3 + 1) (m_1-m_3+ 2)\]
\subsection{Case $k=2$}
In this case $V_3(\R^{3})$ is isomorphic to the group $\rO(3)$. Then  $L^2(V_3(\R^{3}))=L^2(\SO(3))+L^2(\SO(3))$ and 
\[L^2(\SO(3))=\bigoplus_{l\geq 0} W_{2l}\otimes \overline{W}_{2l}\]
The space $W_{l}$ are highest weight $l$ representation of $\so_3(\R)\otimes \C\cong \sl_2(\C)$. The inverse to $t$-transverse part of the metric $\frac{R^2}{4}g_{\Omega}(t)$ is 
\begin{equation}\label{E:symbol}\begin{split}
&\frac{4}{R^2}((1+t)^{-1}\sd_{e_{va}}^2+(1+t)^{-1}\sd_{e_{vb}}^2+(2(1-t))^{-1}\sd_{e_{ab}}^2)=\\
&\frac{4}{R^2}((1+t)^{-1}(\sd_{e_{va}}^2+\sd_{e_{vb}}^2+\sd_{e_{ab}}^2)+\frac{3 t-1}{2 (1-t^2)}\sd_{e_{ab}}^2)
\end{split}\end{equation}
The group $\rO(3)$ is a $\rO(3)\times \SO(2)$-homogeneous space, where $\rO(3)$ acts freely from the left and $\SO(2)$ acts freely from the right. The bi-tensor (\ref{E:symbol}) is a $\rO(3)\times \SO(2)$-invariant. We use the  Levi-Civita connection  associated with a $\rO(3)$-bi-invariant metric $g_b=de_{va}^2+de_{vb}^2+de_{ab}^2$ on $\rO(3)$ to lift the tensor, interpreted as a symbol, to a differential operator. 

We omit  a trivial verification that the resulting operator is 
\begin{equation}
H_{\Omega}=-\frac{4}{R^2(1+t)}\Delta -\frac{2(3 t-1)}{R^2 (1-t^2)}L_{e_{ab}}^2,
\end{equation} where $\Delta$ is the Laplace operator associated with the metric $g_b$ and $L_{e_{ab}}$ is the Lie derivative along the vector field $\sd_{e_{ab}}$.

The operator $\Delta$ up to suitable rescaling coincides with the operator defined by the Casimir element $C \in U(\so_3)$. 
\begin{lemma}\label{L:scalar}
Let $T\in U(\so_{k+1})$ be $\sum_{1\leq s<t\leq k+1}e_{st}e_{st}$. with $e_{st}$ as in (\ref{E:normalization}). The operator $\rho(\Theta)$ in $\R^{k+1}$ acts as a scalar multiplication on $-k/2$.
\end{lemma}

 Let $\rho_{l}$ be representation of $U(\so_3)$ in $W_l$.
 %We conclude that $\Delta$ acts on $W_{2l}\otimes \overline{W}_{2l}$ by $-\frac{l(2l+1)}{3}$. 
 \begin{lemma} The operator $\Delta$ acts on $W_{2l}$ by multiplication on $-\frac{l(2l+1)}{3}$.
 \end{lemma}
 \begin{proof}

It is well known that  $\rho_l(C)=l(l+1)\Id$ (see e.g. \cite{HumphreysRep}). By $\so_{3}$-invariance $T=const\times C$. Note that the $\so_3$ representation $\R^3$ is isomorphic to $W_2$.  We find constant $const$ equal to $-1/6$ by  comparing  $const \times l(l+1)|_{l=2}$ with the scalar $-k/2|_{k=2}$ form Lemma \ref{L:scalar}. Under regular representation of $\so_3$ in $L^2(\SO(3))$ $\rho_{L^2(\SO(3))}(T)=\Delta$. 
 \end{proof}
 \begin{lemma}
 Operator $L_{e_{ab}}^2$ in $W=W_{2l}$  has a spectral decomposition $\bigoplus_{\lambda}W^{\lambda}$. The eigenvalues satisfy  
 \begin{equation}\label{E:eigen-w}
 \lambda= -\frac{s^2}{2},-l\leq s\leq l
 \end{equation}
 For $s=0$ $\dim(W^{0})=1$, for $s\neq 0$ $\dim(W^{-\frac{s^2}{2}})=2$
 \end{lemma}
 \begin{proof}
 The regular element $e_{ab}$ generates a Cartan subalgebra in $\so_3(\R)\otimes \C$. It is known that the operator $\rho_l(e_{ab})$ in $W_l$ can be diagonalized in a weight basis. After a suitable normalization $e_{ab}\leadsto ce_{ab}$ the eigenvalues of $ce_{ab}$ in $W_{2l}$ become $2s$, $-l\leq s\leq l$. We find normalization constant from the condition that $e_{ab}^2$ acts in $\R^3$ by the formula $\rho_{2}(e_{ab}^2)e_a=-1/2e_a,\rho_{2}(e_{ab}^2)e_b=-1/2e_b,\rho_{2}(e_{ab}^2)e_v=0$. On the other hand $ce_{ab}$ has a standard weight basis $e,h,f$ such that $\rho_{2}(ce_{ab})e=2e,\rho_{2}(ce_{ab})h=0,\rho_{2}(ce_{ab})f=-2e$. From this  $c^2=-1/8$ and the eigenvalues are given by  (\ref{E:eigen-w}). The spectral decomposition follows from the weight decomposition 
 \[V^{0}=W^0_{2l}\otimes \overline{W}_{2l},V^{-\frac{s^2}{2}}=(W^{2s}_{2l}+W^{-2s}_{2l})\otimes \overline{W}_{2l}\]

\end{proof}

Our previous results enable us to compute the eigenvalues of $\Delta_{\Omega}$, which are

\[\begin{split}
&\frac{4l(2l+1)}{3R^2(1+t)} +\frac{s^2(3 t-1)}{R^2 (1-t^2)}=\\
&\frac{4 l (2 l+1)}{3 R^2 \left(\sin ^2(\tau )+1\right)}+\frac{s^2 \left(3 \sin ^2(\tau )-1\right)}{R^2 \left(1-\sin ^4(\tau )\right)}+R^2 \left(1-\sin ^2(\tau )\right),-l\leq s\leq l
\end{split}\]

\subsection{Case $k=3$}
In this case $V_3(\R^{4})$ is isomorphic to the group $\SO(4)$. Then  $L^2(V_3(\R^{4}))=L^2(\SO(4))$ and 
\[L^2(\SO(4))=\bigoplus_{l\geq 0,l+m\equiv 0\mod 2} W_{l}\otimes W_{m}\otimes\overline{W}_{l}\otimes \overline{W}_{m}\]

Let $\rho_{reg}U(\so_4)\rightarrow Diff(\SO(4))$ be a identification  of the universal enveloping algebra with the algebra of left-invariant differential operators. The operator $\Delta_{\Omega}$ is equal to image of 
\[-\frac{1}{2}\left(\frac{1}{1+t}e_{12}^2+
\frac{1}{1+t}e_{13}^2+
\frac{1}{2(1-t)}e_{23}^2+
\frac{1}{2t}e_{14}^2+
\frac{1}{1-t}e_{24}^2+
\frac{1}{1-t}e_{34}^2\right)
\]
The eigenvalues  in representations  $W_{2l}\otimes W_{0}$ and $W_{0}\otimes W_{2l}\otimes W_{0}$ are described in the previous subsection.
The eigenvalues  in $W_{1}\otimes W_{1}$ are
\[\left\{\frac{t+5}{4 (t-1) (t+1)},\frac{t+5}{4 (t-1) (t+1)},\frac{3 t+1}{4 (t-1) t},-\frac{5 t+1}{4 t (t+1)}\right\}\]
The formulas for eigenvalues in representations of greater highest weight become significantly more complicates because involve roots of algebraic equations of degree increasing with $l$ and $m$.
For example eigenvalues  in $W_{3}\otimes W_{1}$ are
\[\left\{\frac{3 t^2+12 t+1}{4 (t-1) t (t+1)},\frac{3 t^2+12 t+1}{4 (t-1) t (t+1)},\frac{7 t^2+16 t+1}{4 (t-1) t (t+1)},-\frac{9 t^2-16 t-1}{4 (t-1) t (t+1)},-\frac{t^2+2 \sqrt{13 t^4+4 t^3+2 t^2-4 t+1}-13 t-2}{4 (t-1) t (t+1)},-\frac{t^2+2 \sqrt{13 t^4+4 t^3+2 t^2-4 t+1}-13 t-2}{4 (t-1) t (t+1)},-\frac{t^2-2 \sqrt{13 t^4+4 t^3+2 t^2-4 t+1}-13 t-2}{4 (t-1) t (t+1)},-\frac{t^2-2 \sqrt{13 t^4+4 t^3+2 t^2-4 t+1}-13 t-2}{4 (t-1) t (t+1)}\right\}\]

\subsection{Radial component of the Shr\"odinger operator }
In this section we assume that $k\geq 2$ is an integer. The form $Q(f,g)$ simplifies significantly when $f,g$ are $(e_v,e_a,e_b)$-independent functions  :
\[Q(f,g)= \int_{0}^{1}\left(\frac{4t(1-t)}{R^2}f'(t)\bar g'(t)+R^2(1-t)f(t)\bar g(t)\right)w(t)\d t,\]
$w=w^k_{alg}$ as in (\ref{E:walg}) written in a chart defined by the map $alg$.
We assume that $f,g\in C^{\infty}_c(0,1)$. Integration by parts lead to the  operator $H_{rad}$, which satisfies $Q_{rad}(f,g)=(H_{rad}f,g)_{w}$, where $(f,g)_{w}= \int_{0}^{1}f(t)\bar g(t)w_{alg}(t)\d t$, 
The operator $H_{rad}$ is defined by the formula
\begin{equation}\label{E:schgen}
\begin{split}
&H_{rad}(f)=-w_{alg}^{-1}(p_kf')'+w_{alg}^{-1}q_kf\\
&p_k(t)=\frac{4}{R^2}t(1-t)w_{alg}^{k}(t), q_k(t)=R^2(1-t)w_{alg}^{k}(t)
\end{split}
\end{equation}
For brevity sake we denote $p_k(t)$ and $q_k(t)$ by $p$ and $q$.
The same formula defines an operator, which we denote by the same symbol, in the extended domain $H_{rad}:\O_{an}(\mathbb{C}\backslash\{0,1,-1\})\rightarrow \O_{an}(\mathbb{C}\backslash\{0,1,-1\})$. Here $\O_{an}$ stands for complex analytic functions of parameter $z\in \mathbb{C}\backslash\{0,1,-1\}$. 
The eigenvalue problem $H_{rad}f=\lambda {R^2} f$ in this space becomes an ODE:
\begin{equation}\label{E:sch}
\begin{split}
&f''(t)+\left(\frac{k-1}{t-1}+\frac{k-1}{2 t}+\frac{1}{t+1}\right) f'(t)+\eta^2 \left(\frac{\lambda  -1}{4 t}-\frac{\lambda  }{4 (t-1)}\right)f(t)
 =0,\quad \eta=\frac{R^2}{L} 
\end{split}
\end{equation}
At infinity it has the form
\[
\begin{split}
&f''(t)+\left(\frac{k-1}{t-1}+\frac{5-3 k}{2 t}+\frac{1}{t+1}\right) f'(t)+f(t) \eta^2\left(-\frac{1}{4  t^3}-\frac{\lambda  }{4 t^2}+\frac{\lambda  }{4 (t-1)}-\frac{\lambda  }{4 t}\right)f(t)=0
\end{split}
\]

A more general equation
\[f''(t)+\left(\frac{1-\mu _2}{t-a}+\alpha +\frac{1-\mu _0}{t}+\frac{1-\mu _1}{t-1}\right) f'(t)+\frac{\beta _0+\beta _2 t^2+\beta _1 t}{t (t-1) (t-a)} f(t)=0\]
have been studied in 
It reduced to our equation after  substitution $a\to-1,\alpha\to 0, \mu _0\to \frac{3-k}{2}, \mu _1\to 2-k,\mu _2\to 0, \beta _0\to \frac{1}{4}\eta (1- \lambda ),\beta _1\to -\frac{1}{4} \eta  \lambda ,\beta _2\to -\frac{\eta }{4} $

\subsection{Local solutions}
The method of Frobenius \cite{ince} enables us two find a series solutions at the relevant singular points. 
We find that $z=0$ is a regular singularity with characteristic exponents $(3-k)/2$ and $0$. These series provide us with a general solution of the form

\begin{equation}\label{E:assI}
\begin{split}
&f(z)=c_1y_1(z)+c_2\frac{1}{z^{\frac{k-3}{2}}}y_2(z) ( k\equiv 0\mod 2 ) \text{ or }\\
& f(z)=c_1y_1(z)+c_2\left(\frac{1}{z^{\frac{k-3}{2}}}y_2(z)+\beta(k,R,\lambda)\ln(z)y_1(z)\right)  ( k\equiv 1 \mod 2  )
\end{split}
\end{equation}
where $y_1,y_2$ are analytic functions near zero.

A similar analysis at $z=1$ give characteristic exponents $2-k$ and $0$ and 
\begin{equation}\label{E:assII}
\begin{split}
&f(z)=c_1\tilde{y}_1(z-1)+c_2\left(\frac{1}{(z-1)^{{k-2}}}\tilde{y}_2(z-1)+\alpha(k,R,\lambda)\ln(z-1)\tilde{y}_1(z-1)\right)
\end{split}
\end{equation}
$\tilde{y}_1,\tilde{y}_2$ are analytic at $z=0$

Differential equation (\ref{E:sch}) is equivalent to 
\begin{equation}\label{E:SL}
\begin{split}
&M_k f=-(p_kf')'+q_kf\\
&M_k f-\lambda w^kf=-(p_kf')'+q_kf-\lambda w^k f=0\\
\end{split}
\end{equation}
whose local solutions also have asymptotics (\ref{E:assI},\ref{E:assII}). As usual, when it is clear from the context we abbreviate $M_k$ to $M$.

Let $AC_{loc}(0,1)$ be the space of all complex-valued functions on $(0,1)$ that are absolutely continuous  on  
any compact interval $[\alpha,\beta]\subset (0,1)$.
We define $L^2_w(0,1)$ as a space of Lebesgue measurable functions on $(0,1)$ with a finite $\int_{0}^1|f|^2wdt=(f,f)_w$. 
Let \[D_k=\{f\in AC_{loc}(0,1)| p_kf'\in AC_{loc}(0,1)\}\]
Following \cite{Naimark} we define the maximal domain of $M$ as the subspace
\begin{equation}\label{E:DELTA}
\Delta_k=\left\{f\in D_k|f, \frac{1}{w^k} M_k(f)\in L^2_{w^k}(0,1)\right\}
\end{equation}

\subsection{Endpoint classification of the domain of the  operator $M$}
We follow \cite{Naimark} classification of singularities of (\ref{E:SL}). 
\begin{proposition}\label{P:endpoint}\ \\
The following classification of end-points of the domain $(0,1)$ of the equation (\ref{E:SL}) holds 
\begin{enumerate}
\item  The point $t=0$
	\begin{enumerate}
	\item  is a  regular point  if $k=2$,
	\item \label{b} is limit circle singularity if $k=3,4$,
	\item  \label{c} is limit point singularity if $k\geq 5$.
	\end{enumerate} 
\item  The point $t=1$
	\begin{enumerate}
	\item is limit circle singularity if $k=2$,
	\item is limit point singularity  if $k\geq 3$.
	\end{enumerate}
\end{enumerate}
\end{proposition}
\begin{proof}
The quantities  $|q|,|w|$ are finite on $(0,1)$, therefore only $|p^{-1}|$ can contribute to singularity. The function $p^{-1}$ is never integrable near $t=1$ and $ p^{-1}\in  L^1(0,1/2)$ only when $k=2$.

Expression  $|f|^2(t)w(t)$ is locally integrable  at $0\in (0,1)$ $\forall f$ as in  (\ref{E:assI}) if $k\leq 4$, which takes care of the cases (\ref{b},\ref{c}). 
 
Likewise,   the function $|f|^2(t)w(t)$ is locally integrable  near $1\in (0,1)$  $\forall f$ as in (\ref{E:assII})  if $k\leq 2$.
\end{proof}

\begin{definition}\label{D:domain}
Define the operator $T_k:\Omega_k\rightarrow L^2_w(0,1)$, $k=2,3,\dots$ such that 
\begin{enumerate}
\item $\Omega_2=\{f\in \Delta_2| \lim_{t\rightarrow 0^{+}}(pf')(t)= \lim_{t\rightarrow 1^{-}}(pf')(t)=0\}$
\item $\Omega_3=\{f\in \Delta_3| \lim_{t\rightarrow 0^{+}}(pf')(t)=0\}$
\item $\Omega_4=\{f\in \Delta_4| \lim_{t\rightarrow 0^{+}}(pf')(t)=0\}$
\item $\Omega_k=\Delta_k, k\geq 5$
\end{enumerate}
and $T_kf=\frac{1}{w^k}M_kf$  for $f\in \Omega_k$.
\end{definition}
For self-consistency of this definition the reader should check \cite{BEHZ}.
When confusion can't arise we abbreviate $T_k$ to $T$ and $\Omega_k$ to $\Omega$.
\begin{theorem}
The operator $T$ (see  the definition (\ref{D:domain})) is a self-adjoint in the Hilbert space $L^2_{w}(0,1)$
\end{theorem}
\begin{proof}
See \cite{Naimark} Section 18.
\end{proof}

\begin{theorem}
Let $T$ be the self-adjoint operator as defined in Definition \ref{D:domain}. Then $T$ has the following spectral properties:
\begin{enumerate}
\item The spectrum $\sigma(T)$ is real simple and discrete; that is the spectrum consists solely of real simple eigenvalues , say
\[\sigma(T_k)=\{\lambda_{n}\in \mathbb{R}, n\in \mathbb{N}\}\]
with the property 
\[\begin{split}
&\lambda_{0}\geq 0\\
&\lambda_{n}<\lambda_{n+1}, n\in \mathbb{N}\\
&\lim_{n\rightarrow +\infty}\lambda_{n}=+\infty
\end{split}\]
The operator $T$ is bounded below in $L^2_w(0,1)$ with 
\[(Tf,f)_w\geq 0, f\in \Omega\]
Since each eigenvalue is simple the eigenspaces satisfy 
\[\dim\{f\in \Omega: Tf=\lambda_{n}f\}=1 , n\in \mathbb{N}\]
\item Let the eigenfunction be denoted by $\{\psi_{n}|n\in\mathbb{N}\}$; then $M\psi_{n}=\lambda_{n}w\psi_{n}, n\in \mathbb{N}$
\begin{enumerate}
\item The eigenfunction $\psi_{n}$ satisfy boundary conditions formulated in Definition \ref{D:domain}. 
\item The eigenfunction $\{\psi_{n}|n\in \mathbb{N}\}$ is orthogonal and complete in $L_w^2(0,1)$.
\item For each $n$ $ \psi_{n}$ has exactly  $n$ isolated zeros on $(0,1)$.
\item For each $n$ and $x=0^{+},1^{-}$ the limits $\lim_{t\rightarrow p}\psi_{n}(t)$ exist and finite. The eigenfunctions $\psi_{n}\in AC[0,1]$
\end{enumerate}
\end{enumerate}
\end{theorem}
\begin{proof}
The proof of the similar statement for Heun equation can be found in 
\cite{BEHZ}. It can be used almost verbatim to justify our theorem. Furthermore all other theorems of that paper remain valid in our context.
The key  moment in the proof in \cite{BEHZ}  specific to ordinary Heun    equation are inequalities (83) and (84). Their analogues in our case are the following. Suppose $0\leq t\leq 1/2$ then
\[\begin{split}
&p(t)\geq k_1t^{\frac{k-1}{2}}=\frac{3c_k}{2^{k-2}R^2} t^{\frac{k-1}{2}} \\
&q(t)\leq k_2t^{\frac{k-3}{2}}=c_kR^2t^{\frac{k-3}{2}}\\
&w(t)\leq k_3t^{\frac{k-3}{2}}, k_3 =c_k \text{ if } k\geq 3 \text{ or }  k_3 =3c_k/2  \text{ if } k= 2\\
\end{split}\]
If $1/2\leq t\leq 1$ then 
\[\begin{split}
&p(t)\geq l_1(1-t)^{k-1}=\frac{3c_k}{2^{\frac{k-3}{2}} R^2}(1- t)^{k-1} \\
&q(t)\leq l_2(1-t)^{k-1}, l_2 =2R^2c_k \text{ if } k\geq 3 \text{ or }  l_2 =3c_kR^2/\sqrt{2}  \text{ if } k= 2\\
&w(t)\leq l_3(1-t)^{k-2}, l_3 =2c_k \text{ if } k\geq 3 \text{ or }  l_3 =3c_k/\sqrt{2}  \text{ if } k= 2\\
\end{split}\]

Following arguments of Theorem 23 \cite{BEHZ} that uses Hardy type inequalities we deduce that for $k\neq 3$
\[\int_a^b\left(pf'^2+qf^2-\lambda wf^2\right)\d t> \int_a^b\left(k_1\frac{(k-3)^2}{16}t^{\frac{k-5}{2}}-k_2t^{\frac{k-3}{2}}-k_3|\lambda|t^{\frac{k-3}{2}}\right)f^2 \d t=\]
\[=\int_a^b\left(k_1\frac{(k-3)^2}{16}-(k_2+k_3|\lambda|)t\right)t^{\frac{k-5}{2}}f^2 \d t\]
The last integral is positive if $0<a<b<\frac{k_1(k-3)^2}{16(k_2+k_3|\lambda|)}$. 
If $k=3$ then 
\[\int_a^b\left(pf'^2+qf^2-\lambda wf^2\right)\d t>\int_a^b \left(\frac{k_1}{4t\ln(t)^2}-(k_2+|\lambda|k_3)\right)f^2\d t\]
The integrand is obviously positive if $0<a<b<\delta$ where $\delta$ is sufficiently small. 
Suppose that $1/2<a<b<1$. Then
\[\int_a^b\left(pf'^2+qf^2-\lambda wf^2\right)\d t> \int_a^b\left(l_1\frac{(k-2)^2}{4}(1-t)^{{k-3}}-l_2(1-t)^{{k-1}}-l_3|\lambda|(1-t)^{{k-2}}\right)f^2 \d t=\]
\[=\int_a^b\left(l_1\frac{(k-2)^2}{4}-(l_2(1-t)^2+l_3|\lambda|(1-t))\right)(1-t)^{{k-3}}f^2 \d t\]
If $k=2$ then
\[\int_a^b\left(pf'^2+qf^2-\lambda wf^2\right)\d t>\int_a^b \left(\frac{l_1}{4(1-t)\ln(1-t)^2}-(l_2(1-t)+|\lambda|l_3)\right)f^2\d t=\]
\[=\int_a^b \left(\frac{l_1}{4}-(l_2(1-t)^2\ln(1-t)^2+|\lambda|l_3(1-t)\ln(1-t)^2)\right)\frac{1}{(1-t)\ln(1-t)^2}f^2\d t\]

The rest of the arguments repeats \cite{BEHZ}.
\end{proof}
\subsection{Incorporating  harmonics at $k=2$}
Finally we would like to incorporate harmonics of angular Laplacian.
Our equation becomes
\[-f''(t)-\left(\frac{1}{2 t}+\frac{1}{t+1}+\frac{1}{t-1}\right) f'(t)+f(t) \left(\frac{4 l (2 l+1)}{3 R^2 (1+t)}+\frac{s^2 (3 t-1)}{R^2 (1-t^2)}+\frac{R^2 \left(R^2-\lambda \right)}{4 t}+\frac{\lambda  R^2}{4 (t-1)}\right)=0\]

\[Q_{l,s}(f,g)= c_2\int_{0}^{1}\left(\frac{4t(1-t)}{R^2}f'(t)\bar g'(t)+\left(R^2(1-t)+\frac{4 l (2 l+1)}{3 R^2 (1+t)}+\frac{s^2 (3 t-1)}{R^2 (1-t^2)}\right)f(t)\bar g(t)\right)(1+t)t^{-\frac{1}{2}} \d t,\]

Suppose $0\leq t\leq 1/2$ then
\[\begin{split}
&p(t)\geq k_1t^{\frac{1}{2}}=\frac{3c_2}{R^2} t^{\frac{1}{2}} \\
&q(t)\leq k_2t^{-\frac{1}{2}}=c_2(\frac{4 (2 l+1) l}{3 R^2}+\frac{s^2}{R^2}+R^2)t^{-\frac{1}{2}}\\
&w(t)\leq k_3t^{-\frac{1}{2}},\  k_3 =3c_2/2  \text{ if } k= 2\\
\end{split}\]

\[\int_a^b\left(pf'^2+qf^2-\lambda wf^2\right)\d t>
\int_a^b \left(\frac{k_1}{16}t^{-\frac{3}{2}}-(k_2+|\lambda|k_3)t^{-\frac{1}{2}}\right)f^2\d t=\]
\[=\int_a^b \left(\frac{k_1}{16}-(k_2+|\lambda|k_3)t\right)t^{-\frac{3}{2}}f^2\d t\]

If $1/2\leq t\leq 1$ then 
\[\begin{split}
&p(t)\geq l_1(1-t)=\frac{3\sqrt{2}c_2}{ R^2}(1- t) \\
&q(t)\leq l_2(1-t)^{-1},  l_2 =c_2\frac{\sqrt{2} l (2 l+1)}{3 R^2}+\frac{2 s^2}{R^2}+\frac{3 R^2}{8 \sqrt{2}}  \\
&w(t)\leq l_3, l_3 =3c_2/\sqrt{2} \\
\end{split}\]
\[\int_a^b\left(pf'^2+qf^2-\lambda wf^2\right)\d t>
\int_a^b \left(\frac{l_1}{4(1-t)(\ln(1-t))^2}-(\frac{l_2}{1-t}+|\lambda|l_3)\right)f^2\d t=\]
\[=\int_a^b \left(\frac{l_1}{4}-(l_2+|\lambda|l_3(1-t))(\ln(1-t))^2\right)\frac{1}{(1-t)(\ln(1-t))^2}f^2\d t\]

Additional terms do not change characteristic exponents. The local solutions are 

\begin{equation}\label{E:assI3}
\begin{split}
&f(z)=c_1y_1(z)+c_2z^{\frac{1}{2}}y_2(z)\\
\end{split}
\end{equation}
where $y_1,y_2$ are analytic functions near zero.

\begin{equation}\label{E:assII4}
\begin{split}
&f(z)=c_1y_1(z-1)+c_2\left(y_2(z-1)+\alpha(k,R,\lambda)\ln(z-1)y_1(z-1)\right)
\end{split}
\end{equation}

\subsection{Estimates of the spectral gap of $M$}

R. Lavine in \cite{Lavine} showed that that one-dimensional Schr\"{o}dinger operators with a convex potential admits an estimate for the spectral gap:
\begin{theorem}\label{P:Lavine}
Let $U$ be convex on $[0, r]$, and let $\lambda_1$ and $\lambda_2$ be the first two eigenvalues for the Dirichlet Schr\"{o}dinger operator $-d^2/dx^2 + U$ on $[0, r]$. Then $\lambda_2-\lambda_1 \geq \Gamma_0$ where $\Gamma_0$ is the gap for constant $U$ for the Dirichlet  operator with equality only if $U$ is constant. Thus
\[\lambda_2-\lambda_1\geq \frac{3\pi}{r^2}\]
\end{theorem}
We are going to use this theorem to find the spectral gap for $H_{rad}$ $k\geq 5$. For this purpose we  rewrite 
bilinear form (\ref{E:qshrodinger}) in the chart associated with the map $trig$ and restrict it to the space of $\SO(k+1)$-invariant functions. We get 
\[Q_{trig}(f,g)=\int_{0}^{\pi R/2}\left(f'\bar g'+R^2\cos^2\left(\tau/R\right)f\bar g\right)w\d \tau\] 
where $w=w_{trig}$ as in (\ref{E:trig}). A substitution $f\rightarrow fw^{-1/2}$ transform $Q_{trig}(f,g)$ to 
\[E_{trig}(f,g)=\int_{0}^{\pi R/2}\left(f'\bar g'+V_{eff}(\tau)f\bar g\right)\d \tau.\]  The substitution simultaneously trivializes the inner product $\int_{0}^{\pi R/2}f\bar g w\d \tau\leadsto \int_{0}^{\pi R/2}f\bar g \d \tau$.
The differential equation associated  is 
\begin{equation}
N_k(f)=-f''+V_{eff,k}f=\lambda f.
\end{equation}
Here we introduce $k$-dependence of the potential $V_{eff}$, which was implicit in the previous formulas of this section. The reader can find a formula for $V_{eff}$ and discussion of its properties in Appendix \ref{S:poly}. The given formulation of the spectral problem is unitary equivalent to the formulation given in Definition  \ref{D:domain} if we define 
 \[\begin{split}
 &D'(a,b)=\{f\in AC_{loc}(a,b)| f'\in AC_{loc}(a,b)\}\quad  \Delta'_k(a,b)=\left\{f\in D'(a,b)|f, N_k(f)\in L^2(a,b)\right\},0\leq a<b\leq \pi R/2\\
 &D'=D'(0,\pi R/2)\quad  \Delta'_k=\Delta'_k(0,\pi R/2)\}
 \end{split}\]
\begin{definition}\label{D:domain}
Define the operator $S_k:\Omega'_k\rightarrow L^2(0,\pi R/2)$, $k=5,6,\dots$ such that  $\Omega'_k=\Delta'_k, $
and $S_kf=N_kf$  for $f\in \Omega_k$.
\end{definition}
Unfortunately Lavine's theorem is valid only for a regular Sturm-Liouville. 
Following \cite{Zettl} we approximate  a singular Sturm-Liouville problem  on an interval $(0,\pi R/2)$ by a sequence of regular Sturm-Liouville problems on
truncated intervals $(a_r,b_r)$ where
\[0 < a_r < b_r < \pi R/2\quad  a_r\geq a_{r+1}, a_r\rightarrow 0 \quad b_{r+1}\geq b_r, b_r\rightarrow \pi R/2 \text{ as } r\rightarrow \infty\]
By $S_r=S_{r,k}$ we denote self-adjoint
realizations of $N_k$ on the intervals $(a_r, b_r)$. Thus 
$S_r$ are self-adjoint operators in the Hilbert spaces $H_r =
L^2((a_r,b_r))$. The spectrum and eigenvalues are denoted
by $\sigma(S_r), \lambda_n(S_r), n\in \mathbb{N}$. Among possible $S_r$ there are  "inherited" operators $S^i_r$ that are  defined by "inherited"
boundary conditions. Here is an adaptation of the general Definition 10.8.1 from \cite{Zettl}  to our situation.
\begin{definition}\label{D:domanetr}
  The domain of the self-adjoint $S^i_{r,k}$ $k\geq 5$ is
  $\Omega'_{r,k}=\{f\in \Delta'_k(a_r,b_r)| \lim_{t\rightarrow a_r^{+}}f(t)= \lim_{t\rightarrow b_r^{-}}f(t)=0\}$ -the Dirichlet condition. 
\end{definition}

We will also need the following approximation theorem adapted to our needs.
\begin{theorem} \label{T:cont}\cite{BEWZ} 
Suppose
\[My = —(py')' + qy = \lambda wy \text{ on } J = (a, b), — \infty < a < b < \infty, \]
under the conditions
\[\frac{1}{p},q,w \in  L_{loc}(J,\mathbb{R}), w > 0 \text{ a.e. on } J \]
in the Hilbert space $H = L^2(J,w)$
and 
\[-\infty < a < a_r < b_r < b < \infty\]
hold. Let $S$ be a self-adjoint realization of $(M,w)$ on $(a,b)$ where $a,b$ limit point singularity of $M$. Let $S^i_r$ be the
inherited operator on $(a_r,b_r)$. Assume that the spectrum $\sigma(S)=\{\lambda_{n}(S)|n\in \mathbb{N}\}$ is bounded below and discrete. 
Then
\[\lambda_n(S^i_r) \rightarrow  \lambda_n(S), \text{ as }r\rightarrow \infty , \text{ for each } n\in \mathbb{N}\]
\end{theorem}

\begin{corollary}

The operator  $S_k$ (Definition \ref{D:domain}) has a discrete simple real spectrum. The eigenvalues $\lambda_1$ and $\lambda_2$ satisfy $\lambda_2-\lambda_1\geq \frac{12}{R^2}$ if $k\geq 5$ and $R<59049 (k-4) (k-2)/4096$.
\end{corollary}
\begin{proof}
The operators $S^i_r$ (see Definition \ref{D:domanetr}) define a regular Sturm-Liouville problem. They have discrete spectrum (see e.g. \cite{Zettl}). Let $\Gamma (A)$ be the difference $\lambda_2(A)-\lambda_1(A)$ .    By Theorem \ref{T:cont} the inequality $\Gamma(S_n^i)\geq \frac{3\pi}{(b_n-a_n)^2}$ implies the statement of corollary. Convexity of $V_{eff}$ has been verified in Proposition \ref{P:convex}. Therefore corollary follows from Theorem \ref{P:Lavine}.
\end{proof}

The same form in $\tau$-coordinate is 
\[
Q(f,g)-\lambda(f,g)= \int_{0}^{\pi R/2}\left(f'\bar g'+(R^2\cos^2(\tau/R)-\lambda)f\bar g\right)  w\d \tau 
\]
The fraction $\frac{Q(f,f)}{(f,f)}$ for function $f(x)=1$ is equal to $\frac{1 - 2 k}{1 - 3 k}R^2$
we conclude that the first eigenvalue satisfies the inequality \[0\leq \lambda_1\leq \frac{1 - 2 k}{1 - 3 k}R^2\]

We transform the form to Liouville form by a change of variables $f=w^{-\frac{1}{2}}F$, where $w= \sin^{k-2}\left(\frac{\tau}{R}\right)\cos^{2k-3}\left(\frac{\tau}{R}\right)(1+\sin^{2}\left(\frac{\tau}{R}\right))$.
The result is 
\[Q(F,G)=\frac{R^{3k-3}}{ 2^{\frac{5k-4}{2}}}  \int_{0}^{\frac{\pi R}{2}}F'(\tau)\bar G'(\tau)+V(\tau)F(\tau)\bar G(\tau)\d \tau\]
Normalization of the inner product is $\frac{R^{3k-3}}{ 2^{\frac{5k-4}{2}}} \int_{0}^{\frac{\pi R}{2}} F(\tau)\bar G(\tau)\d \tau$
We set $V_{eff}=V+R^2\cos^2(\tau/R)-\lambda$. Then $V_{eff}=\frac{A(\csc(\tau/R))}{B(\csc(\tau/R))}$, where $A$ and $B$ are the following polynomials:
\begin{equation}
\begin{split}
&A(x)=\left(k^2-6 k+8\right) x^{10}+  \left(-4 k^2+12 k+4 R^4-4 \lambda  R^2-6\right)x^8+\\
&\left(-2 k^2-8 k-4 \lambda  R^2+5\right)x^6 + \left(12 k^2-44 k-8 R^4+4 \lambda  R^2+44\right)x^4+\\
& \left(9 k^2-18 k+4 \lambda  R^2+9\right)x^2+4 R^4\\
&B(x)=4 R^2 x^2 \left(x^2-1\right) \left(x^2+1\right)^2
\end{split}
\end{equation}
Function $V_{eff}(\tau)$ has poles at $z=\frac{\pi R}{2}s$, $i R\log \left(1+\sqrt{2}\right)+\pi R s$, $iR \log \left(\sqrt{2}-1\right)+\pi R s, s\in \Z$.
\[\begin{split}
&V_{eff}\sim \frac{k^2-6 k+8}{4 (\tau-\pi R s) ^2 }+\dots,\mu=2-k/2,k/2-1\\
&V_{eff}\sim \frac{4 k^2-16 k+15}{4 \left(\tau -(\pi/2-\pi s)R\right)^2 }+\dots,  \mu=5/2-k,k-3/2\\
&V_{eff}\sim -\frac{1}{4  \left(\tau -R(i \log \left(1+\sqrt{2}\right)+\pi s)\right)^2}+\dots, \mu=1/2\\
&V_{eff}\sim -\frac{1}{4  \left(\tau -R(i \log \left(\sqrt{2}-1\right)+\pi s)\right)^2}+\dots \mu=1/2
\end{split}\]
In the above formulas $\mu$ are the characteristic exponents of the singularity.
Local solutions for various values of $k$ $\tau'=\tau-\pi R/2$
\begin{equation}
\begin{split}
&{\bf k=2}\\
&\tau y_1(\tau^2)\\
& y_2(\tau^2)\\
& \tau'^{1/2}y_1(\tau'^2), \\
&\ln(\tau')\tau'^{1/2} y_1(\tau'^2)+\tau'^{5/2}y_2(\tau'^2)\\
&{\bf k=3}\\
& \tau^{1/2}y_1(\tau^2), \\
&\ln(\tau) \tau^{1/2}y_1(\tau^2)+\tau^{5/2}y_2(\tau^2), \\
& \tau'^{3/2}y_1(\tau'^2), \\
&\tau'^{3/2}(\ln(\tau') y_1(\tau'^2)+\tau'^{-1/2}y_2(\tau'^2)\\
&{\bf k= 4}\\
&y_1(\tau^2)\\
&\tau y_2(\tau^2)\\
&\tau'^{5/2}y_1(\tau'^2)\\
&\ln(\tau')\tau'^{5/2}y_2(\tau'^2)+\tau'^{-3/2}y_1(\tau'^2)\\
&{\bf k\geq 5}\\
&\tau^{k/2-1}y_1(\tau^2)\\
&\ln(\tau)\tau^{k/2-1}y_1(\tau^2)+\tau^{2-k/2}y_2(\tau^2)\\
&\tau'^{k-3/2}y_1(\tau'^2)\\
&\ln(\tau')\tau'^{k-3/2}y_2(\tau'^2)+\tau'^{5/2-k}y_1(\tau'^2)\\
\end{split}
\end{equation}
If $k\geq 5$ and $R$ is sufficiently small the function $V$ is convex. From this we conclude that 
\[\lambda_2-\lambda_1\geq \frac{12}{R^2} \]

\end{document}